%% file: jumplists.tex
	\documentclass[
		paper=a4,%
		11pt,american,pagesize,%
		DIV=12,headinclude=true,headlines=0,footinclude=true,footlines=-5,twoside=semi%
	]{scrartcl}
	\def\version{arxiv}

\def\draftmode{false}

\usepackage{hyperref}

\input{jumplists-preamble}

\makeatletter
\wref@putname{sec}{\S\!}
\wref@putname{fig}{Fig.\!}
\wref@putname{tab}{Tab.\!}
\wref@putname{thm}{Thm.\!}
\wref@putname{lem}{Lem.\!}

\def\mytitle{Median-of-$k$ Jumplists and Dangling-Min BSTs}
\def\myfunding{%
	The last author is supported by the 
	Natural Sciences and Engineering Research Council of Canada 
	and the Canada Research Chairs Programme.
}

\ifsiam{
	\title{\mytitle\thanks{\strut\myfunding}}
}{
	\title{\boldmath\mytitle\thanks{\myfunding}}
}

	\author{%
		Markus E.\ Nebel%
		\thanks{%
			Technische Fakultät, Universität Bielefeld, Germany.\protect\\
			\texttt{nebel\,@\,techfak.uni-bielefeld.de}%
		}
	\and
		Elisabeth Neumann%
		\thanks{%
			Carl-Friedrich-Gauß-Fakultät, 
			\protect\\
			Technische Universität Braunschweig, Germany.
			\protect\\
			\texttt{e.neumann\,@\,tu-braunschweig.de}%
		}
	\and 
		Sebastian Wild%
		\thanks{%
			David R.\ Cheriton School of Computer Science,
			University of Waterloo, Canada.
			Email: \texttt{wild\,@\,uwaterloo.ca}%
		}
	}

\ifsiam{
	\fancyfoot[R]{\scriptsize{Copyright \textcopyright\ 2019\\
	 Copyright for this paper is retained by authors.}}
}{}

\date{\small\today}

\begin{document}

\maketitle

\begin{abstract}
\ifsubmission{%
	~\\\noindent\textbf{\textsf{Abstract:}\;}\\
}{}
We extend randomized jumplists introduced by Brönnimann, Cazals, and Durand~\cite{BronnimannCazalsDurand2003}
to choose jump-pointer targets as median of a small sample for better search costs,
and present randomized algorithms with expected $\Oh(\log n)$ time complexity 
that maintain the probability distribution of jump pointers 
upon insertions and deletions.
We analyze the expected costs to search, insert and delete a random element,
and we show that omitting jump pointers in small sublists hardly affects search costs,
but significantly reduces the memory consumption.

We use a bijection between jumplists and ``dangling-min BSTs'', 
a variant of (fringe-balanced) binary search trees for the analysis.
Despite their similarities, some standard analysis techniques 
for search trees fail for dangling-min trees (and hence for jumplists).
\end{abstract}

\section{Introduction}

Jumplists were introduced by Brönnimann, Cazals, and Durand~\cite{BronnimannCazalsDurand2003}
as a simple randomized comparison-based dictionary implementation.
They allow iteration over the stored elements in \emph{sorted} order
and supports queries and updates in expected logarithmic time.
The core is a sorted (\mbox{singly-})\,linked listed augmented with 
\emph{jump pointers}, \ie, shortcuts that speed up searches.
Jump-pointers are required to be well-nested, \ie, they may not cross.
This allows binary-search-like navigation.
\wref{fig:typical-jumplist-n30-k1-w2} shows an exemplary jumplist;
a detailed definition is deferred to~\wref{sec:jumplist-definitions}.

\begin{figure}
	\plaincenter{\includegraphics[width=.9\linewidth]{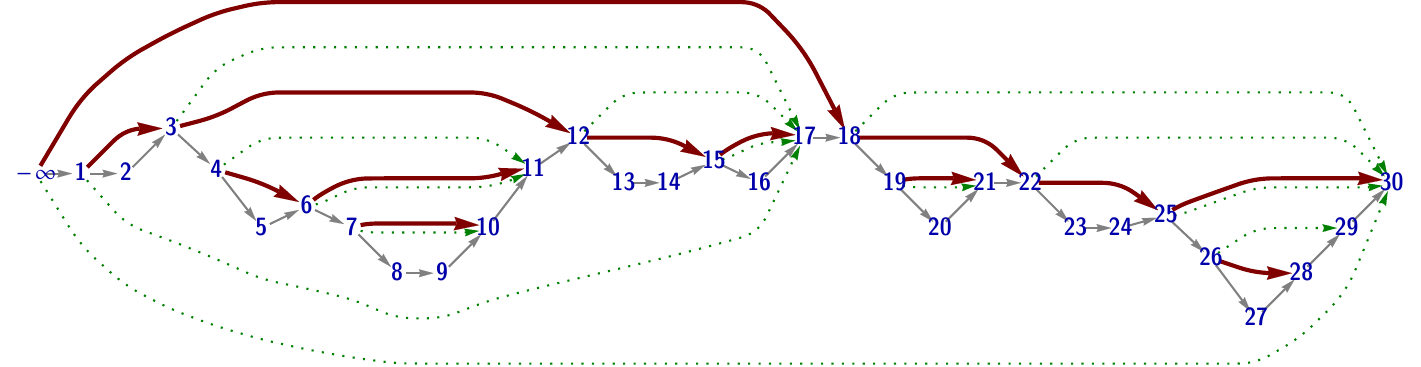}}
	\caption{%
		A jumplist on $n=30$ keys (with $k=1$ and $w=2$).
		Gray arrows are backbone links, thick red arrows are jump pointers.
		Dotted green arrows delimit a node's conceptual sublist; (they are
		not stored).%
	}
	\vspace{-\smallskipamount}
	\label{fig:typical-jumplist-n30-k1-w2}
\end{figure}

If all jump pointers point to the middle of their sublist,
we obtain perfect binary search,
but we need a rule that is also efficiently maintainable upon insertions and deletions.
Brönnimann, Cazals, and Durand~\cite{BronnimannCazalsDurand2003} proposed a \emph{randomized} solution: 
jump pointers invariably have a uniform distribution over their sublist,
\ie, the first jump pointer equally likely points to any element
and thereby divides the list in two parts, the next- and jump-sublists.
Both follow the same rule recursively;
since pointers may not cross, they can do so independently.

In this article, we generalize jumplists to use a more balanced distribution: 
each jump pointer points to the \emph{median of a small sample of $k$ elements} 
of its sublist.
(The original jumplists correspond to $k=1$.)
Building on the algorithms from~\cite{BronnimannCazalsDurand2003} 
we present $\Oh(\log n)$ expected-time insertion and deletion algorithms 
for median-of-$k$ jumplists that maintain this more balanced distribution.
Here $n$ counts the number of keys currently stored.
A larger $k$ balances the structure more rigidly which improves searches,
but makes the cleanup after updates more expensive.
Our main contribution is an analysis of median-of-$k$ jumplists that 
precisely quantifies the influence of $k$ on searches, insertions and deletions. 

We also introduce a novel search strategy (named \emph{spine search}) 
that reduces the number of needed key comparisons significantly,
and we suggest a further modification of jumplists: 
for sublists smaller than a threshold $w$, we omit the jump pointers altogether.
This allows to trade space for time: 
elements in these small sublists do not have to store a jump pointer, 
but the corresponding subfile can only be searched sequentially.
We show that this saves a constant fraction of the pointers
while affecting expected search costs only by an additive constant.

\paragraph{Outline of the paper}
In the remainder of the introduction we
summarize related work.
\wref{sec:preliminaries} contains common notation and preliminaries used later.
In \wref{sec:jumplist-definitions}, we define jumplists. 
We present our spine search strategy in \wref{sec:spine-search}.
\wref{sec:median-of-k-jumplists-def} introduces the median-of-$k$ extension,
and \wref{sec:insert-delete} describes the insertion and deletion algorithms.
Our analysis is given in \wref{sec:analysis}, 
and we conclude the paper with a discussion of the results (\wref{sec:conclusion}).
\ifsiam{%
	There is an \extendedversion of this paper available as arxiv preprint
	that contains detailed descriptions of all algorithms
	and some missing proofs.%
}{%
	The appendix contains a list of used notations,
	as well as details on the operations and omitted parts of the analysis.%
}

\subsection{Related Work}
\label{sec:related-work}
(Unbalanced) binary search trees (BSTs) perform close to optimal 
on average and with high probability 
when keys are inserted in random order~\cite{Knuth1998,Mahmoud1992evolution}.
A standard approach is to enforce the average behavior through randomization.
The most direct application of this paradigm is given by
Martínez and Roura~\cite{MartinezRoura1998} who devised efficient 
randomized insert and delete operations 
that maintain the shape distribution of random insertions.
The idea also works when duplicate keys are allowed~\cite{Pasanen2010}.

Randomized BSTs store subtree sizes for maintaining the distribution.
The \emph{treaps} of Seidel and Aragon~\cite{SeidelAragon1996} instead store a 
random priority with each node. 
Treaps remain in random shape by enforcing a heap order \wrt the random priorities.
Their performance characteristics are very similar to randomized BSTs.

Unless further memory is used, BSTs do not offer $O(1)$ time successor queries.
Like jumplists, Pugh's skip lists~\cite{Pugh1990} are augmented, sorted linked lists, 
so successors are found by following one pointer.
Skip lists extend the list elements by towers of pointers of different heights, where
each tower cell points to the successor among all element of at least this height.
With geometrically distributed heights, operations run in $\Oh(\log n)$ expected time 
with $\Oh(n)$ extra pointers in expectation.
The varying tower heights can be inconvenient;
this originally motivated the introduction of jumplists.
For skip lists, there is a direct and transparent bijection to BSTs~\cite{DeanJones2007}; 
this becomes more complicated for jumplists
(see \wref{sec:jumplist-definitions}).

The classic alternative to randomization are deterministically balanced BSTs~\cite{AnderssonFagerbergLarsen2005}.
Munro, Papadakis, and Sedgewick~\cite{MunroPapadakisSedgewick1992} transfer 
the height-balance rule of 2-3 trees to skip lists, and
Elmasry~\cite{Elmasry2005} applied the weight-balancing criterion of 
$\mathit{BB}[\alpha]$ trees~\cite{NievergeltReingold1973} to jumplists.
Note that the latter achieves logarithmic update time only in an amortized sense.

A constant-factor speedup over BSTs is achieved with \emph{fringe-balanced} BSTs.
The name originates from \emph{fringe analysis}, 
a technique used in their analysis~\cite{PobleteMunro1985}.%
\footnote{
	The concept
	appears under a handful of other names in the (earlier) literature:
	\emph{locally balanced search trees}~\cite{Walker1976},
	\emph{diminished trees}~\cite{Greene1983}, and
	\emph{iR / SR trees}~\cite{HuangWong1983,HuangWong1984}.
}
In a fringe-balanced search tree, 
leaves \emph{collect} keys in a buffer.
Once a leaf holds $k$ keys, it is \emph{split:}
the median of the $k$ elements is used as 
the key of a new node;
two new leaves holding the other elements form its subtrees.
Many parameters like expected path length, height and profiles of fringe-balanced 
trees have been studied~\cite{Drmota2009}.

\section{Notation and Preliminaries}
\label{sec:preliminaries}

\ifsiam{%
	We first introduce some notation.%
}{%
	We introduce some important notation here;
	\wref{app:notations} gives a comprehensive list.%
}
We use Iverson's bracket $[\mathit{stmt}]$ to mean $1$ if $\mathit{stmt}$ is true and $0$ otherwise.
Falling resp.\ rising factorial powers are denoted by $x^{\underline n}$ and $x^{\overline n}$;
for negative $n$ holds 
$x^{\underline n} = 1 / (x+1)^{\overline n}$ resp.\
$x^{\overline n} = 1 / (x+1)^{\underline n}$.
$\Prob{E}$ denotes the probability of event $E$ and $\E{X}$ the expectation
of random variable $X$. We write $X\eqdist Y$ to denote equality in distribution.

For a self-contained presentation, 
we list here a few mathematical preliminaries used in the analysis later.

\paragraph{Beta  distribution}
The \emph{beta distribution} has two parameters $\alpha,\beta\in\R_{>0}$ 
and is written as $\betadist(\alpha,\beta)$.
If  $X\eqdist\betadist(\alpha,\beta)$, we have $X \in (0,1)$ and $X$ has the density 
\begin{align*}
		f(x)
	&\wrel=
		\frac{x^{\alpha-1}(1-x)^{\beta-1} }{\BetaFun(\alpha,\beta)}
		,\qquad x\in(0,1),
\end{align*}
where $\BetaFun(\alpha,\beta) = \Gamma(\alpha)\Gamma(\beta) / \Gamma(\alpha+\beta)$
is the beta function.

The following lemma is helpful for 
computing expectations involving such beta distributed variables;
it is a special case of
\cite[\href{https://www.wild-inter.net/publications/html/wild-2016.pdf.html\#pf60}{Lemma 2.30}]{Wild2016}.
\begin{lemma}[``Powers-to-Parameters'']
\label{lem:powers-to-parameters}
	Let $X_1$ be a \(\betadist(\alpha_1,\alpha_2)\)
	distributed random variable and write $X_2=1-X_1$.
	Let further \(m_1,m_2 \in \Z^d\)
	with \(m_1,m_2 > -\alpha\) be given and abbreviate
	\(A \ce \alpha_1 + \alpha_2\) and \(M \ce m_1+m_2\). 
	Then for an arbitrary (real-valued, measurable) function~\(f\) 
	holds
	\begin{align*}
			\E[\big]{ X_1^{m_1} X_2^{m_2} \cdot f(X_1) }
		&\wwrel=
			\frac{\alpha_1^{\overline{m_1}} \alpha_2^{\overline{m_2}}}
				{A^{\overline M}} 
			\cdot \E[\big]{ f(\tilde X_1) }
			\,,
	\end{align*}
	where \(\tilde X_1\) is
	\(\betadist(\alpha_1+m_1,\alpha_2+m_2)\) distributed.
\qed\end{lemma}

\paragraph{Beta-Binomial Distribution}
The \emph{beta-bi\-no\-mi\-al distribution} is a discrete distribution 
with parameters $n\in\N_0$ and $\alpha,\beta \in \R_{>0}$.
It is written as $\betaBinomial(n,\alpha,\beta)$.
If $I\eqdist\betaBinomial(n,\alpha,\beta)$, we have $I\in[0..n]$ and 
\begin{align*}
		\Prob{I = i}
	&\wrel=
		\binom{n}{i} \frac{\BetaFun(\alpha+i,\beta+(n-i))}{\BetaFun(\alpha,\beta)}
		,\quad i \in \Z \;.
\end{align*}
(Recall that $\binom ni$ is zero unless $i\in[0..n]$.)
An alternative representation of the weights for $\alpha = t_1+1,\beta = t_2+1 \in\N$ with $k=t_1+t_2+1$ is
\begin{align*}
		\binom{n}{i} \frac{\BetaFun(\alpha+i,\beta+(n-i))}{\BetaFun(\alpha,\beta)}
	&\wrel=
		\frac{\binom{i+t_1}{t_1}\binom{n-i+t_2}{t_2}}{\binom{n+k}{k}},
\label{eq:dirichlet-multinomial-weights-binomial}
\end{align*}
which yields a combinatorial interpretation.
There is a second way to 
obtain beta-binomial distributed random variables: 
we first draw a random probability $D \eqdist \betadist(\alpha,\beta)$ according
to a beta distribution, and then use this as the success probability of a binomial distribution,
\ie, $I \eqdist \binomial(n; d)$ \emph{conditional} on $D = d$.
The beta-binomial distribution is thus also called a \emph{mixed} binomial distribution,
using a beta-distributed \emph{mixer} $D$;
this explains its name.

Since the binomial distribution is sharply concentrated, one can
use Chernoff bounds on beta binomial variables after conditioning on the beta distributed
success probability.
That already implies that $\betaBinomial(n,\alpha,\beta)/n$ converges to $\betadist(\alpha,\beta)$
(in a specific sense).
We can obtain the stronger error bounds given in the following lemma 
by directly comparing the probability density functions.

\savebox\citehrefbox{%
\bfseries\cite[\href{https://www.wild-inter.net/publications/html/wild-2016.pdf.html\#pf66}{Lem.\,2.38}]{Wild2016}%
}

\begin{lemma}[{Local limit law~%
	\usebox\citehrefbox%
}]
\label{lem:limit-law-beta-binomial}
	Let $(\ui In)_{n\in\N}$ be a sequence of random variables where
	$\ui In$ is distributed like $\betaBinomial(n,\alpha,\beta)$ for $\alpha,\beta\in\N_{\ge1}$.
	Then for $n\to\infty$ we have uniformly for $z\in(0,1)$ that 
	\begin{align}
			n \Prob[\big]{\ui Jn /n \in (z-\tfrac1n,z] }
		&\wwrel=
			f_B(z) \bin\pm \Oh(n^{-1}),
	\end{align}
	where $f_B(z) = z^{\alpha-1}(1-z)^{\beta-1} / \BetaFun(\alpha,\beta)$ is the density
	function of the beta distribution with parameters $\alpha$ and~$\beta$.
\qed\end{lemma}
	Since $f_B$ is a polynomial in $z$, it is in particular bounded and Lipschitz continuous
	in the closed domain $z \in [0,1]$.
	Hence, the local limit law also holds for the random variables $\ui{J}n = \ui{I}{n-d} + c$ for constants $c$ and $d$.
Further properties of the beta-binomial distribution are collected
in~\cite[\href{https://www.wild-inter.net/publications/html/wild-2016.pdf.html\#pf64}{\S\,2.4.7}]{Wild2016}.

\ifsiam{%
	We list the following expectations here for reference.
	The proofs are simple computations found in the \extendedversion.%
}{%
	The following expectations are listed here for reference;
	proofs are given in \wref{app:expectations}.%
}
\begin{lemma}
	\label{lem:binomial-negative-factorial-moments}
	Let $X \eqdist \binomial(n,p)$ for $n\in\N_0$ and $p\in(0,1]$. 
	Then we have with $q=1-p$ that
	\begin{align*}
			\E*{X^{\underline{-1}}} 
		&\wrel=
			n^{\underline{-1}} \cdot p^{-1} (1-q^{n+1})\,,
	\\
			\E*{X^{\underline{-2}}} 
		&\wrel\le
			n^{\underline{-2}} \cdot p^{-2}\;.
	\end{align*}
\end{lemma}
\begin{lemma}
\label{lem:E-ln-D}
For $D\eqdist \betadist(t+1,t+1)$ we have (with $k=2t+1$)
\begin{align*}
		\E{\ln D}
	&\wwrel=
		\harm{t} - \harm{k}
		,
\\
		\E{D \ln D}
	&\wwrel=
		\frac12\bigl(\harm{t+1} - \harm{k+1}\bigr)
		.
\end{align*}
\end{lemma}
\paragraph{Hölder continuity}
A function $f:I\to \R$ defined on a bounded interval $I$ is 
Hölder continuous with exponent $h\in(0,1]$
when 
\[
	\exists C\;
	\forall x,y\in I\wrel:
		\bigl| f(x) - f(y) \bigr|
		\wrel\le 
		C |x-y|^h.
\]
Hölder continuity is a notion of smoothness 
that is stricter than (uniform) continuity, but slightly more liberal
than Lipschitz continuity (which corresponds to $h=1$).
$f:[0,1]\to\R$ with $f(z) = z \ln(1/z)$ is a stereotypical function
that is Hölder continuous (for any $h\in(0,1)$), but not Lipschitz.

For functions defined on a bounded domain,
Lipschitz continuity implies Hölder continuity
and Hölder continuity with exponent $h$ implies
Hölder continuity with exponent $h' < h$.
Recall that a real-valued function is Lipschitz 
if its derivative is bounded.

\subsection{The Distributional Master Theorem}
\label{app:DMT}

To solve the recurrences in \wref{sec:analysis},
we use the ``distributional master theorem'' (DMT)%
~\cite[\href{https://www.wild-inter.net/publications/html/wild-2016.pdf.html\#pf8a}{Thm.~2.76}]{Wild2016}, 
reproduced below for convenience.
It is based on Roura's continuous master theorem~\cite{Roura2001}, 
but reformulated in terms of distributional recurrences in an attempt to give 
the technical conditions and occurring constants in Roura's original formulation 
a more intuitive, stochastic interpretation.
We start with a bit of motivation for the latter.

The DMT is targeted at divide-and-conquer recurrences where the recursive parts
have a \emph{random} size. 
The average-case analyses of Quicksort and binary search trees are typical examples 
that lead to such recurrences.
Because of the random subproblem sizes, a traditional recurrence for expected costs
has to sum over all possible subproblem sizes, weighted appropriately.
That way, the direct correspondence between the recurrence and the algorithmic process 
is lost, in particular the number of recursive applications is no longer 
directly visible.

An alternative that avoids this is a \emph{distributional recurrence}
that describes the full distribution of costs.
The distribution for larger problem sizes is described by a 
``toll term'' (for the divide and/or combine step) 
plus the contributions of recursive applications.
Such a distributional formulation requires 
the toll costs and subproblem sizes to be stochastically 
independent of the recursive costs when conditioned on the subproblem sizes.
In typical applications, this is fulfilled when the studied algorithm 
guarantees that the subproblems on which it calls itself recursively
are of the same nature as the original problem.
Such a form of randomness preservation is also required for the analysis
using traditional recurrences.
We can thus use the distributional language to 
describe costs directly mimicking the structure of our algorithms in this paper.

The DMT allows us to compute an asymptotic approximation of the expected costs
directly from the distributional recurrence.
Intuitively speaking, it is applicable whenever the \emph{relative} subproblem sizes
of recursive applications converge to a (non-degenerate) limit distribution as $n\to\infty$
(in a suitable sense; see \weqref{eq:DMTwc-condition} below).
The local limit law provided by \wref{lem:limit-law-beta-binomial}
gives exactly such a limit distribution.

\savebox\citehrefbox{%
\bfseries\cite[\href{https://www.wild-inter.net/publications/html/wild-2016.pdf.html\#pf8a}{Thm.~2.76}]{Wild2016}%
}
\begin{theorem}[DMT \usebox\citehrefbox]
\label{thm:DMTwc}
	Let \((C_n)_{n\in\N_0}\) be a family of random variables that
	satisfies the distributional recurrence
	\begin{align}
	\label{eq:DMTwc-distributional-recurrence}
			C_n
		\wwrel\eqdist
			T_n \bin+ \sum_{r=1}^s \ui{A_r}n \cdot C_{\ui{J_r}n}^{(r)},
			\qquad (n \ge n_0),
	\end{align}
	where the families \((\ui{C_n}1)_{n\in\N},\ldots,(\ui{C_n}s)_{n\in\N}\) are independent copies of
	\((C_n)_{n\in\N}\), which are also independent of 
	\((\ui{J_1}n,\ldots,\ui{J_s}n)\in\{0,\ldots,n-1\}^s\),
	\((\ui{A_1}n,\ldots,\ui{A_s}n) \in \R_{\ge0}^s\)
	and \(T_n\).
	Define \(\ui{Z_r}n = \ui{J_r}n / n\), $=1,\ldots,s$, and assume that they 
	fulfill uniformly for \(z\in(0,1)\)
	\begin{align}
	\label{eq:DMTwc-condition}
			n \cdot \Prob[\big]{ \ui{Z_r}n \in (z-\tfrac1n,z] }
		&\wwrel= 
			f_{Z_r^*}(z) \bin\pm \Oh(n^{-\delta}),
	\end{align}
	as $n\to\infty$ for a constant \(\delta>0\) and a 
	Hölder-continuous function \(f_{Z_r^*} : [0,1] \to \R\).
	Then \(f_{Z_r^*}\) is the density of a random variable \(Z_r^*\) and
	\(\ui{Z_r}n \convD Z_r^*\).

	Let further 
	\begin{align}
	\label{eq:DMTwc-condition-coeffs}
			\E[\big]{ \ui{A_r}n \given \ui{Z_r}n \in (z-\tfrac1n,z] }
		&\wwrel=
			a_r(z) \wbin\pm \Oh(n^{-\delta}),
	\end{align}
	as $n\to\infty$ for a function \(a_r : [0,1] \to \R\) and
	require that \(f_{Z_r^*}(z)\cdot a_r(z)\) is also Hölder continuous on~\([0,1]\).
	Moreover, assume \(\E{T_n} \sim K n^\alpha \log^\beta(n)\), as \(n\to\infty\), 
	for constants \(K\ne 0\), \(\alpha\ge0\) and \(\beta>-1\).
	Then, with \(H = 1 - \sum_{r=1}^s\E{(Z_r^*)^\alpha a_r(Z_r^*)}\), we have the following cases.
	\begin{enumerate}[itemsep=0ex]
	\item If \(H > 0\), then \(\displaystyle \E{C_n} \sim \frac{\E{T_n}}{H}\).
			\(\vphantom{\displaystyle\sum_{i}^{i}}\)
	\item \label{case:DMTwc-H0} 
		If \(H = 0\), then 
		$\displaystyle
		\E{C_n} \sim \frac{\E{T_n} \ln n}{\tilde H}$ with 
		$\displaystyle \tilde H = -(\beta+1)\sum_{r=1}^s 
			\E{(Z_r^*)^\alpha a_r(Z_r^*) \ln(Z_r^*)}$.
	\item \label{case:DMTwc-theta-nc}
		If \(H < 0\), then \(\E{C_n} = \Oh(n^c)\) for the
		\(c\in\R\) with \(\displaystyle\sum_{r=1}^s\E{(Z_r^*)^c a_r(Z_r^*)} = 1\).
	\end{enumerate}
\qed\end{theorem}

\section{Jumplists}
\label{sec:jumplist-definitions}

We now present our (consolidated) definition of jumplists;
\ifsiam{%
	some details differ from the original~\cite{BronnimannCazalsDurand2003};
	we discuss those in the \extendedversion.
}{%
	it deviates in some details from the original version of \cite{BronnimannCazalsDurand2003};
	see \wref{app:differences-definitions}.
}

Jumplists consist of \emph{nodes}, where each node $v$ stores a successor pointer ($v.\id{next}$)
and a key ($v.\id{key}$).
The nodes are connected using the next pointers to form a singly-linked list, 
the \emph{backbone} of the jumplist, so that the key fields are sorted ascendingly.%
\footnote{%
	We assume the keys stored in a jumplist are distinct.
	The insert procedures will prevent duplicate insertions.
}
It is convenient to add a ``dummy'' header node $v_0$ whose key field is ignored;
($v_0.\id{key} = -\infty$).
If $x_1 < \cdots < x_n$ are the keys stored in the jumplist,
we have the $n+1$ nodes $v_0,v_1,\ldots,v_n$ with
$v_i.\id{key} = x_i$ and $v_{i-1}.\id{next} = v_i$ for $i=1,\ldots,n$.
A jumplist on $n$ keys will always have $m=n+1$ nodes;
we use $n$ and $m$ in this meaning throughout the paper.

\paragraph{Jump Pointers}
Jump pointers always point forward in the list, 
and we require the following two conditions.
\begin{inlineenumerate}%
\item
	\emph{Non-degeneracy:} Any node may be the target of at most one jump pointer,
		and jump pointers never point to the direct successor.
\item 
	\emph{Well-nestedness:} 
	Let $v\ne u$ be nodes with $v.\mathit{key} < u.\mathit{key}$,
	and let $v^*$ resp.\ $u^*$ be the nodes their jump pointers point to.
	(Note that $v^*\ne u^*$ by the first property).
	Then these nodes must appear in one of the following orders in the backbone:
		$u \dots v\dots v^* \dots u^*$
		or 
		$v \dots v^* \dots u \dots u^*$:
	\\[.5\baselineskip]
	\plaincenter{\adjustbox{max width=.9\linewidth}{%
	\begin{tikzpicture}[
		]
		\begin{scope}
			\foreach \x/\l/\p in {u/u/1,v/v/2,{v^{\mkern-1mu*}\!}/v2/3.5,{u^{\mkern-1mu*}\!}/u2/4.5} {
				\node[sn,minimum size=4.25mm] (\l) at (\p,0) {$\x{}^{\vphantom *}$} ;
			}
			
			\draw[jumppointer] (u) to[out=60,looseness=.7] (u2) ;
			\draw[jumppointer] (v) to[out=60,looseness=.7] (v2) ;
		\end{scope}
		\node[scale=1.5] at (5.75,0) {\normalfont or} ;
		\begin{scope}[shift={(6,0)}]
			\foreach \x/\l/\p in {u/u/1,v/v/3.25,{v^{\mkern-1mu*}\!}/v2/4.5,{u^{\mkern-1mu*}\!}/u2/2.25} {
				\node[sn,minimum size=4.25mm] (\l) at (\p,0) {$\x{}^{\vphantom *}$} ;
			}
			
			\draw[jumppointer] (u) to[out=60,looseness=1.2] (u2) ;
			\draw[jumppointer] (v) to[out=60,looseness=1.2] (v2) ;
		\end{scope}
	\end{tikzpicture}%
	}}\\[.5\baselineskip]
	The second case allows $v^* = u$.
	Visually speaking, jump pointers may not cross.
\end{inlineenumerate}

\paragraph{Sublists}

The \emph{sublist of node $v$} starts at $v$ (inclusive) and 
ends just before the first node 
targeted by a jump pointer originating before $v$~-- 
or extends to the end of the list if no overarching pointer exists.
As for the overall jumplist,
$v$ acts as dummy header to its sublist:
$v.\id{key}$ is \emph{not} considered as part of $v$'s sublist.
We write $m(v)$ for the number of nodes in $v$'s sublist.
The next- and jump-sublists of $v$,
denoted by $\mathcal J_1 = \mathcal J_1(v)$ resp.\ $\mathcal J_2 = \mathcal J_2(v)$,
are the sublists of $v.\id{next}$ resp.\ $v.\id{jump}$.
We use $J_r=J_r(v)$ for the number of nodes in $\mathcal J_r(v)$, $r\in\{1,2\}$.
\wref{fig:def-sublists} exemplifies the definitions.
We include an imaginary ``end pointer'' in the figures, drawn as dotted green line,
that connects a jump node with the last node in that node's sublist.

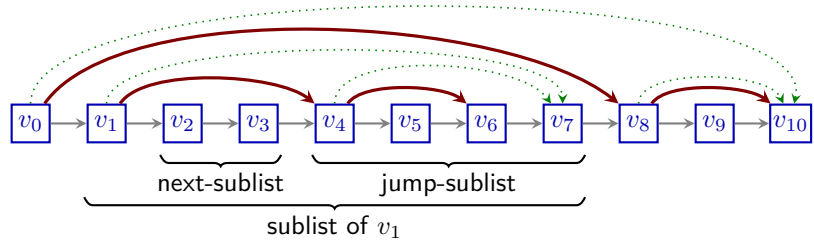
\begin{figure}
	\plaincenter{
	\adjustbox{max width=\linewidth}{%
		\begin{tikzpicture}[
			yscale=.8,
			node distance=1cm,
		]
			\node[sn] (0) {$v_0$};
			\foreach \i/\l in {1/0,2/1,3/2,4/3,5/4,6/5,7/6,8/7,9/8,10/9} {
			    \node[sn] (\i) [right of=\l]  {$v_{\i}$};
			}
			
			\foreach \i/\n in {0/1,1/2,2/3,3/4,4/5,5/6,6/7,7/8,8/9,9/10} {
				\draw[backbone] (\i) -- (\n) ;
			}
			\foreach \i/\n in {0/8,1/4,4/6,8/10} {
				\draw[jumppointer] (\i) to[out=60,looseness=.7] (\n) ;
			}
			\draw[endpointer] (0) to[out=90,in=80,looseness=.55] (10) ;
			\draw[endpointer] (1) to[out=90,in=90,looseness=.5] (7) ;
			\draw[endpointer] (4) to[out=90,in=120,looseness=.8] (7) ;
			\draw[endpointer] (8) to[out=90,in=110,looseness=.8] (10) ;
			
			\begin{scope}[thick, decoration={ brace, mirror}, shorten >=-2pt,shorten <=-2pt]
				\draw[decorate] ($(2.180) + (0,-.6)$) -- node[below=.2ex]{next-sublist} ++($(3.360)-(2.180)$);
				\draw[decorate] ($(4.180) + (0,-.6)$) -- node[below=.2ex]{jump-sublist} ++($(7.360)-(4.180)$);
				\draw[decorate] ($(1.180) + (0,-1.3)$) -- node[below=.2ex]{sublist of $v_1$} 
					++($(8.360)-(2.180)$);
			\end{scope}
		\end{tikzpicture}%
		}%
	}
	\caption{%
		Illustration of the sublist definitions.
		The sublist of node $v_1$ contains $m(v_1)=7$ nodes and
		stores the $6$ keys $v_2.\mathit{key},\ldots,v_7.\mathit{key}$.
		The sizes of the next- and jump-sublist 
		are $J_1(v_1) = 2$ and $J_2(v_1)=4$, respectively.%
	}
	\label{fig:def-sublists}
\end{figure}

\paragraph{Node Types}
Nodes in our jumplists come in two flavors:
\emph{plain nodes} only have next and key fields;
\emph{jump nodes} additionally store a 
\emph{jump pointer}, $v.\id{jump}$, and their next-sublist size, $v.\id{nsize}=J_1$.
The node types are determined by the following rule,
where $w\ge2$, the \emph{leaf size,} is a parameter:
If $m(v)\le w$, then $v$ (and all nodes in its sublist) are plain nodes.
Otherwise $v$ is a jump node, and we apply the rule recursively
to $\mathcal J_1(v)$ and $\mathcal J_2(v)$.
\wref{fig:typical-jumplist-n30-k1-w2} shows a larger example.

\paragraph{Randomized Jumplists}

The following probability distribution over all (legal) jump-pointer configurations
invariantly holds in randomized jumplists.
It is defined recursively: $v_0.\mathit{jump}$ is drawn \emph{uniformly}
from all $m-2$ feasible targets;
($v_0$ and $v_1$ are not allowed).
Conditional on the choice of $v_0.\mathit{jump}$, 
the same property is required independently for $\mathcal J_1(v_0)$ and $\mathcal J_2(v_0)$.
The probability $p(\mathcal J)$ of a particular (legal) pointer configuration $\mathcal J$ is
\begin{align*}
\label{eq:prob-jumplist-k1}
		p(\mathcal J)
	&\wwrel=
		\begin{dcases*}
			1, & $m \le w$; \\
			\frac1{m-2} \cdot p(\mathcal J_1) \,p(\mathcal J_2), & 
				$m > w$,
		\end{dcases*}
\end{align*}
which is reminiscent of the probability of a given shape for a random BST,
except for the offset $-2$
(see \cite[ex.\,6.2.2--5]{Knuth1998} or \cite[Eq.\,(5.1)]{CasasDiazMartinez1991}).

\subsection{Dangling-Min BSTs}
\label{sec:dangling-min-bsts}

There is an intimate relation between jumplists and search trees, 
but the slight offset above complicates the matter.%
\footnote{%
	The complication is inherent to the feature of jumplists that
	every key has at most \emph{one} jump pointer.
	Skip lists, for example, can be transformed into BSTs directly~\cite{DeanJones2007}.
}
Indeed, (random) jumplists are isomorphic to a rather peculiar variant of (random) BSTs
(where random means ``generated by insertions in random order''):
the \emph{dangling-min BSTs} (with leaf size $w\ge 2$).
Such a tree is defined for a sequence of (distinct) keys $x_1,\ldots,x_n$ as follows.
	If $n\le w-1$, it is a leaf 
	with the keys in sorted order.
	Otherwise, its \emph{root} node contains \emph{two} keys:
	the smallest key, $\min\{x_1,\ldots,x_n\}$, as its \emph{dangling min,}
	and the first key of the sequence after the min has been removed as \emph{root key}
	(\ie, the root key is $x_1$, unless $x_1$ is the min; then it is~$x_2$).
	The left resp.\ right subtrees of the root are the dangling-min BSTs
	for the keys smaller resp.\ larger than the root key
	in the remaining sequence (without root key and min, and preserving relative order).
Dangling-min BSTs make the recursive decomposition in jumplists explicit,
which helps for both designing algorithms and analyzing their performance.

\begin{figure}
	\plaincenter{\adjustbox{max width=\linewidth}{%
		\begin{tikzpicture}[
			scale=.9,
			node distance=.7cm,
		]
			\node[minBST internal] (11) at (0,0) {$11$} ;
			\node[minBST internal] (5) at ($(11)+(-1.1,-1)$) {$5$} ;
			\node[minBST leaf] (12) at ($(11)+(1.1,-1)$) {$12$} ;
			\node[minBST internal] (4) at ($(5)+(-1.25,-1)$) {$4$} ;
			\node[minBST internal] (10) at ($(5)+(1.25,-1)$) {$10$} ;
			\node[minBST internal] (8) at ($(10)+(-.6,-1)$) {$8$} ;
			\node[minBST leaf empty] (35) at ($(4)+(-.5,-1)$) {} ;
			\node[minBST leaf empty] (45) at ($(4)+(.5,-1)$) {} ;
			\node[minBST leaf empty] (105) at ($(10)+(.6,-1)$) {} ;
			\node[minBST leaf] (9) at ($(8)+(.5,-1)$) {$9$} ;
			\node[minBST leaf empty] (75) at ($(8)+(-.5,-1)$) {} ;
			\foreach \f/\s in {11/5,5/4,10/8,8/75,4/35} {
				\draw[minBST left edge] (\f) -- (\s) ;
			}
			\foreach \f/\s in {11/12,5/10,10/105,8/9,4/45} {
				\draw[minBST right edge] (\f) -- (\s) ;
			}
			\foreach \r/\m in {11/1,5/2,4/3,10/6,8/7} {
				\node[minBST min] (m\r) at ($(\r)+(195:1.75em)$) {$\m$};
				\draw[minBST min edge] (\r) -- (m\r) ;
			}
			
			\begin{scope}[
				shift={(2.15,-3.5)},every node/.append style={font=\scriptsize}
			]
				\node[sn,minimum size=4mm] (0) {};
				\foreach \i/\p in {1/0,2/1,3/2,4/3,5/4,6/5,7/6,8/7,9/8,10/9,11/10,12/11} {
					\node[sn,right of=\p,minimum size=4mm] (\i) {$\i$} ;
				}
				
				\foreach \i/\n in {0/11,1/5,5/10,6/8,2/4} {
					\draw[jumppointer] (\i) to[out=60,looseness=.7] (\n) ;
				}
				\path[backbone]
					(0)  edge (1)
					(1)  edge (2)
					(2)  edge (3)
					(3)  edge (4)
					(4)  edge (5)
					(5)  edge (6)
					(6)  edge (7)
					(7)  edge (8)
					(8)  edge (9)
					(9)  edge (10)
					(10) edge  (11)
					(11) edge  (12)
				;
			\end{scope}
		\end{tikzpicture}%
		}%
	}
	\caption{%
		The dangling-min BST with $w=2$ for the sequence $11,2,5,3,1,4,10,8,7,9,6,12$,
		and the jumplist it corresponds to.
	}
	\label{fig:example-min-BST}
\end{figure}
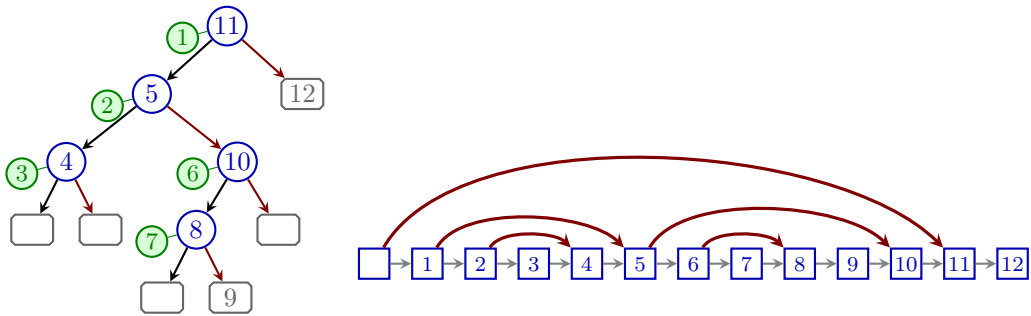

We can transform a jumplist to a dangling-min BST (and vice versa):
If $m\le w$, $v_0$ is a plain node and the dangling-min BST
is a leaf containing all $m-1\le w-1$ keys;
(recall that a jumplist with $m$ nodes stores $n=m-1$ keys).
Otherwise, $v_0$ is a jump node;
with $x_1$ the key in $v_0.\id{next}$ and $x_j$ the key in $v_0.\id{jump}$,
the root of the dangling-min BST has root key $x_j$ and dangling min $x_1$.
Next- resp.\ jump-sublist are recursively transformed into left and right subtree.
\wref{fig:example-min-BST} shows the jumplist corresponding to the given tree;
\wref{fig:typical-jumplist-tree-n30-k1-w2} gives a larger example.

It is easy to see inductively 
that the dangling-min BST built from a randomized jumplist
has the same distribution as 
if directly constructed for a random permutation of $\{1,\ldots,n\}$.
We can therefore focus on analyzing the latter.

\section{Spine Search}
\label{sec:spine-search}

Searching a key $x$ in a jumplist is straightforward: 
We start at the header. We stop when the key in the current node $v$ is larger or equal to $x$.
Otherwise we follow either the jump pointer~-- if the key in $v.\id{jump}$ is not larger than $x$~--
or the next-pointer.
We call this strategy the classic search in the sequel.%
\footnote{%
	Brönnimann, Cazals, and Durand~\cite{BronnimannCazalsDurand2003} also studied
	the symmetric alternative
	\parenthesisclause{compare first to $v.\id{next}$ and then with $v.\id{jump}$ (if needed)}
	and found that it needs more comparisons on average.
}

However, there is an alternative search strategy not considered 
in~\cite{BronnimannCazalsDurand2003} and~\cite{Elmasry2005},
which performs better!
Consider searching key $8$ in the jumplist from
\wref{fig:typical-jumplist-n30-k1-w2}.
A classic search in this list inspects keys $18,1,3,12,4,6,11,7,10,8$ in the given order;
a total of $10$ key comparisons. 
Every step in the search that follows the next-pointer needs two comparisons.

Now do the search for $8$ in the dangling-min BST from \wref{fig:typical-jumplist-tree-n30-k1-w2},
as if it was a regular BST 
(ignoring the subtree minima and stopping at the leaves).
While doing so, we compare with keys $18,3,12,6,11,10$. 
All these steps need only one key comparison even though mostly the same keys are 
visited as above.
However, our search is not yet finished; the reached leaf contains only $9$, 
and we would (erroneously!) announce that $8$ is not in the dictionary.
Instead we have to return to the \emph{last node we entered through a right-child pointer} and 
inspect all the dangling mins along the ``left spine'' of the corresponding subtree.
In our example, we return to $11$ and make comparisons with $7$ and $8$, 
terminating successfully.
We call this search strategy \emph{spine search.} 
In our example, it needed $2$ comparisons less than the classic search.

\begin{figure}
	\plaincenter{\adjustbox{max width=.66\linewidth}{%
		\includegraphics[scale=.6]{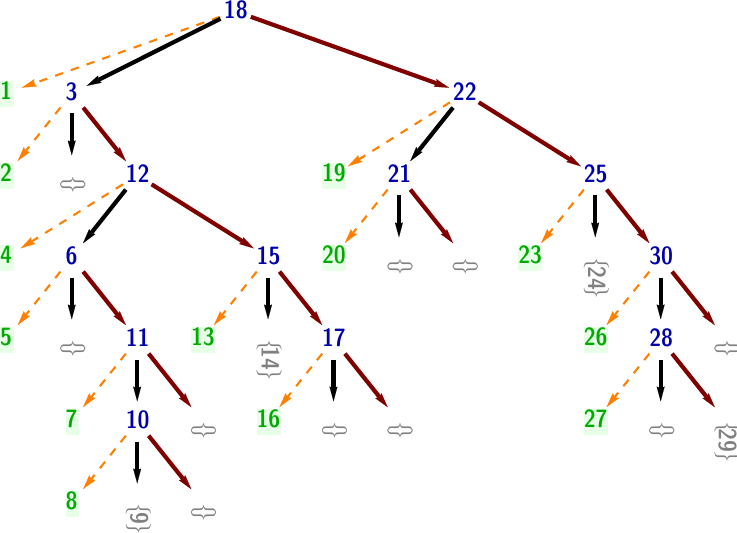}%
	}}
	\caption{%
		The dangling-min BST for the jumplist from \wref{fig:typical-jumplist-n30-k1-w2}.
		Black arrows are left child pointers, red arrows are right child pointers, 
		and dotted yellow arrows indicate the dangling min.
		Gray nodes are leaves that contain between $0$ and $w-1=1$ keys.
	}
	\label{fig:typical-jumplist-tree-n30-k1-w2}
\end{figure}

Spine search only compares $x$ with the dangling-mins for nodes on the 
\emph{left spine above the leaf,}
whereas the classic strategy does so for \emph{every} node we leave through the left-child edge.
Our modification is correct because when going to the right child 
we know that all keys left to $v$ are smaller than $x$ and thus $x$ cannot 
be any of the dangling minima we skipped.
\ifsiam{%
	The \extendedversion%
}{%
	\wref{app:algorithms}%
}
gives detailed pseudocode.

The left spine is always a subset of the nodes where we took a left child edge,
so spine search never needs more comparisons than the classic strategy.
It seems reasonable that spine search should need roughly as many key comparisons as the search 
in a BST since most left spines are short.
Indeed, we prove in \wref{sec:analysis} that the linear search along the left spine 
is only a lower order term when averaging over all possible unsuccessful searches \parenthesissign
spine search needs $\sim 2\ln(n)$ comparisons,
compared to $\sim3\ln(n)$ for the classic search strategy.

\section{Median-of-k Jumplists}
\label{sec:median-of-k-jumplists-def}

The search costs in BSTs can be improved
by using medians of a small sample as subtree roots;
the idea is called fringe-balancing in that context (\wref{sec:related-work})
and corresponds to the median-of-$k$ rule for Quicksort~\cite{Hennequin1991,Drmota2009,Wild2018}.
Applied to our trees, we obtain 
\emph{$k$-fringe-balanced dangling-min BSTs:} if $n\ge w$, we choose the root key
as the median of the first $k$ keys in the sequence after removing the min
(and otherwise proceed as before).
Here $k = 2t+1$ is a fixed odd integer and we require $w\ge k+1$.

Similarly, we define a \emph{randomized median-of-$k$ jumplist} 
by choosing the jump target as the median of $k$ elements.
The situation is illustrated below for $k=3$ and $m=10$;
to have $x_6$ as the median of $3$ elements from the sample range,
we must select $t=1$ further elements from $\{x_2,\ldots,x_5\}$ and 
$t=1$ further elements from $\{x_7,\ldots,x_9\}$.

\medskip

\plaincenter{
\adjustbox{max width=\linewidth}{%
	\begin{tikzpicture}[
		scale=.75
	]
		\fill[rounded corners=3pt,blue!15] (1.5,.5) rectangle (9.5,-.5)  
			node[fill,text=blue, inner sep=1.5pt,align=center,anchor=south] at (7.5,.4) 
				{\smaller sampling range};
		
		\foreach \i in {0,...,9} {
			\node[sn,semithick,scale=.75] (\i) at (\i,0) {$x_{\i}$};
		}
		\foreach \i in {6,4,9} {
			\draw[ultra thick,<-,red] (\i) -- +(0,-.8) ;
		}
		\node[red] at (10,-.7) {sample};
		\draw[jumppointer] (0) to[out=45,in=120,looseness=.4] (6) ;
		
		\begin{scope}[decoration={brace,mirror},thick]
			\draw[decorate] ($(1.west) + (0,-1)$) -- 
				node[below=.5ex]{$\mathcal J_1$} ++($(5.east)-(1.west)$) ;
			\draw[decorate] ($(6.west) + (0,-1)$) -- 
				node[below=.5ex]{$\mathcal J_2$} ++($(9.east)-(6.west)$) ;
		\end{scope}
		
	\end{tikzpicture}%
}
}

\noindent
The number of such samples is $\binom {J_1-1}t \binom {J_2-1}t$, which we have to divide by
the total number of possible samples, $\binom {m-2}k$.
The probability of a (legal) jump pointer configuration $\mathcal J$ thus is
\begin{align*}
\label{eq:prob-jumplist-general-k}
		p(\mathcal J)
	&\wrel=
		\begin{dcases*}
			1 & $m \le w$; \\
			\frac{\binom{J_1-1}{t} \binom{J_2-1}{t}}{\binom{m-2}k} 
				\cdot p(\mathcal J_1) \,p(\mathcal J_2), & 
				$m > w$.
		\end{dcases*}
\end{align*}
This puts more probability weight on balanced configurations, 
and hence improves the expected search costs.
\wref{fig:typical-jumplist-n30-k3-w4} shows a typical median-of-$3$ jumplist
and its fringe-balanced dangling-min tree.%
\footnote{%
	A possible generalization could use asymmetric sampling with
	$(t_1,t_2)$ and $k=t_1+t_2+1$, where we select the $(t_1+1)$st smallest instead
	of the median. 
	Then, we have $\binom{J_1-1}{t_1}$ and $\binom{J_2-1}{t_2}$ 
	in \weqref{eq:prob-jumplist-general-k}.
	For the present work, we will however stick to the case $t_1=t_2=t$.
}

\begin{figure}
	\plaincenter{\adjustbox{max width=\linewidth}{
		\hspace*{-1em}%
		\raisebox{4ex}{\includegraphics[scale=.6]{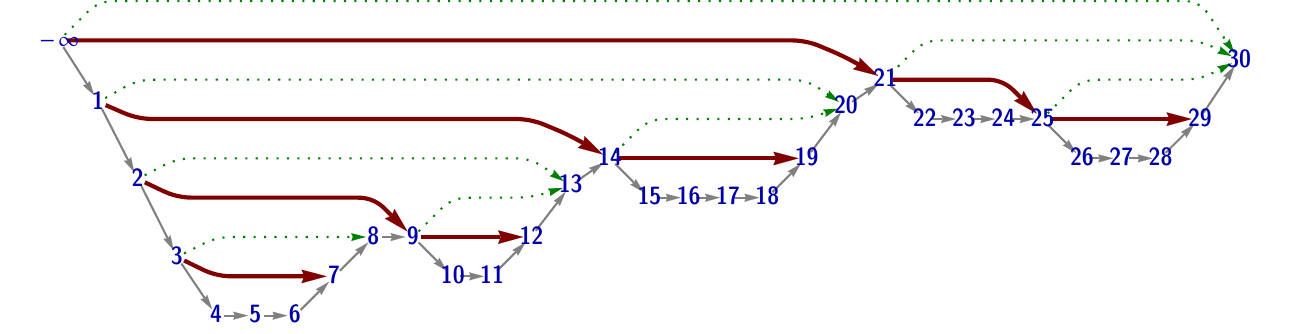}}%
		\hspace*{-1.75em}%
	}}
	\\[-6ex]
	\plaincenter{\adjustbox{max width=.6\linewidth}{%
		\includegraphics[scale=.6]{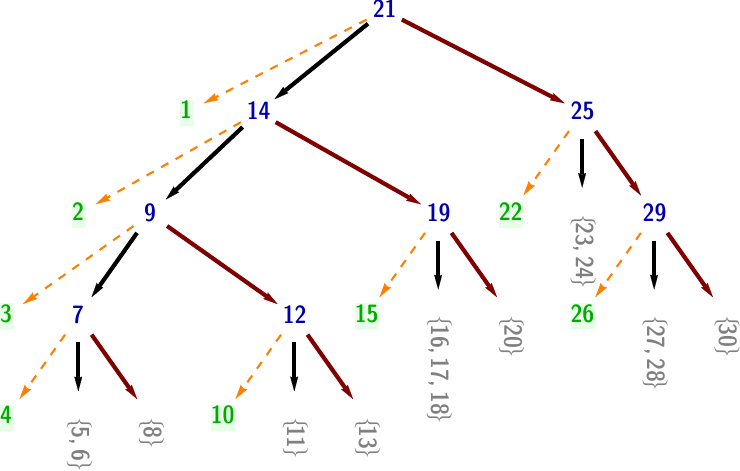}%
	}}%
	\caption{%
		A typical median-of-three ($k=3$, $w=4$) jumplist on $n=30$ keys
		and its corresponding fringe-balanced dangling-min BST.
	}
	\label{fig:typical-jumplist-n30-k3-w4}
\end{figure}

\paragraph{Distribution of subproblem sizes}
For our analysis, an alternative description of the distribution of the subproblem sizes
is more convenient.
Note that both $J_1$ and $J_2$ are always at least $t+1$: 
the sublists must contain $t$ other sampled nodes plus their header.
If we denote by $I_r = J_r-t-1$, $r\in\{1,2\}$, we find that $I_r$ 
has a beta-binomial distribution (\wref{sec:preliminaries}), 
$I_r \eqdist \betaBinomial(m-2-k,t+1,t+1)$.
This implies that with $D \eqdist \betadist(t+1,t+1)$, 
we have the mixed distribution $I_r \eqdist \binomial(m-2-k,D)$ conditional on~$D$.%
\footnote{%
	The symmetry in the sublist sizes, $J_1\eqdist J_2$,
	is a major convenience of our definition of jumplists 
	as opposed to the original one.
}

\section{Insert and Delete}
\label{sec:insert-delete}

\begin{figure*}
	\savebox\tmpbox{
	\begin{tikzpicture}[
			node distance=1cm,
			every node/.style={inner sep=1.5pt,anchor=west},
		]
		\node (head) at (-1.5,0) {$\proc{RestIns}(\Li)$};
		\path (head.east) --node[anchor=center] {$=$} ++(3em,0) 
				node (p) {$p \cdot {}$} ;
				
		\node[sn] (v0) at (p.east) {$v_0$};
		\node (3)  [right of = v0] {$\proc{Reb} \biggl($};
		\node[sn] (J1) at (3.east) {\phantom{xxxxxxx}}; 
		\node (4)  at (J1.east){$\biggl)$ };
		
		\node (5) [xshift=1em] at (4.east) {$\proc{Reb} \biggl($};
		\node[sn] (6) at (5.east){$v_{J}$};
		\node[sn,xshift=-.8pt] (J2) at (6.east) {\phantom{xxxxxxxxx}}; 
		\node (7) [color=nodecolor,anchor=east]  at (J2.east) {$v_n\,$};
		\node at (J2.east) {$\biggr)$} ;
		
		\path[backbone,on layer=background]
			(v0) edge[shorten >=-2pt] (3)
			(4) edge[shorten >=-2pt,shorten <=-4pt] (5);
		\draw[jumppointer,overlay]
			(v0) to[out=45,looseness=.4,in=90] (6);
	
		\begin{scope}[shift={(1.5,-2.75)}]
			\node (p1)  {${}+(1-p)\cdot {}$};

			\node[sn,xshift=.8em] (n2) at (p1.east) {$v_0$};
			\node (n3)  [right of = n2, xshift=.7em] {$\proc{RestIns} \biggl($};
			\node (n4) [new,xshift=1.2em] at (n3.east) {};
			\node[sn] (nJ1) at (n3.east) {\phantom{xxxxxxx}}; 
			\node (n5) at (nJ1.east){$\biggl)$ };
			\node[sn] (n6) [right of = n5, xshift=-1em]  {$v_j$};
			\node[sn,anchor=west,xshift=-.6pt] (nJ2) at (n6.east) {\phantom{xxxxxx}};

			\node[sn,yshift=8.5ex] (a2) at (n2.west) {$x$};
			\node (a3)  [right of = a2, xshift=.7em] {$\proc{RestIns} \biggl($};
			\node  (a4) [new,xshift=0em] at (a3.east) {};
			\node[sn] (aJ1) at (a3.east) {\phantom{xxxxxxx}}; 
			\node (a5) at (aJ1.east){$\biggl)$};
			\node[sn] (a6) [right of = a5, xshift=-1em]  {$v_j$};
			\node[sn,xshift=-.8pt] (aJ2) at (a6.east) {\phantom{xxxxxx}};

			\node[sn,yshift=-8.5ex] (b2) at (n2.west) {$v_0$};
			\node[sn,xshift=.7em] (bJ1) at (b2.east) {\phantom{xxxxxxx}}; 
			\node[xshift=.7em] (b3) at (bJ1.east) {$\proc{restIns} \biggl($};
			\node[sn] (b4) at (b3.east) {$v_j$};	
			\node[new,xshift=1em] (b5)  at (b4.east) {};
			\node[sn,anchor=west,xshift=-.8pt] (bJ2) at (b4.east) {\phantom{xxxxxx}}; 
			\node (b6) at (bJ2.east){$\biggl)$};
			
			\path[backbone,on layer=background]
				(n2) edge[shorten >=-2pt] (n3)
				(n5) edge[shorten <=-4pt] (n6)
				(a2) edge[shorten >=-2pt] (a3)
				(a5) edge[shorten <=-4pt] (a6)
				(b2) edge (bJ1)
				(bJ1) edge[shorten >=-2pt] (b3);
			\draw[jumppointer] (n2) to[out=45,looseness=.4,in=120] (n6);
			\draw[jumppointer] (a2) to[out=45,looseness=.4,in=120] (a6);
			\draw[jumppointer] (b2) to[out=45,looseness=.45,in=90] (b4);
		
			\draw[semithick,decorate, decoration={ brace}] 
				($(b2.south west)+(-.5em,-1ex)$) -- ($(a2.north west)+(-.5em,1ex)$);	
		\end{scope}
			
		\begin{scope}[shift={(10,-2.75)}, node distance = 0.8cm]
			\node (cm1) {if $\Li={}$};
			\node (cm2) [sn] at (cm1.east) {$v_0$};
			\node (cm3) [new,xshift=2.2em] at (cm2.east) {};
			\node[sn,xshift=1em] (cmJ1) at (cm2.east) {\phantom{xxxxxxx}};
			\node[sn,xshift=1em] (cm4) at (cmJ1.east) {$v_j$};
			\node[sn,anchor=west,xshift=-0.8pt] (cmJ2) at (cm4.east) {\phantom{xxxxxx}};

			\node[yshift=8.5ex] (ca1) at (cm1.west) {if $\Li={}$};
			\node[sn,fill=red] (ca2)  at (ca1.east) {$x$};
			\node[sn,xshift=.7em,minimum width=4mm] (ca3) at (ca2.east) {$v_0$};
			\node[sn,xshift=.7em] (caJ1) at (ca3.east) {\phantom{xxxx}};
			\node[sn,xshift=.7em] (ca4) at (caJ1.east) {$v_j$};
			\node[sn,anchor=west,xshift=-0.8pt] (caJ2) at (ca4.east) {\phantom{xxxxxx}};

			\node[yshift=-8.5ex] (cb1) at (cm1.west) {if $\Li={}$};
			\node[sn] (cb2) at (cb1.east) {$v_0$};
			\node[sn,xshift=1em] (cbJ1) at (cb2.east) {\phantom{xxxxxxx}};
			\node[sn,xshift=1em] (cb3) at (cbJ1.east) {$v_j$};
			\node[new,xshift=1em] (cb4) at (cb3.east) {};
			\node[sn,anchor=west,xshift=-0.8pt] (cbJ2) at (cb3.east) {\phantom{xxxxxx}};

			\path[backbone,on layer=background]
				(cm2) edge (cmJ1.west)
				(cmJ1.east) edge (cm4)
				(ca2) edge (ca3)
				(ca3) edge (caJ1)
				(caJ1) edge (ca4)
				(cb2) edge (cbJ1.west)
				(cbJ1.east) edge (cb3);
			\draw[jumppointer] (cm2) to[out=45,looseness=.55,in=120] (cm4);
			\draw[jumppointer] (ca3) to[out=50,looseness=.7,in=120] (ca4);	
			\draw[jumppointer] (cb2) to[out=45,looseness=.5,in=120] (cb3);	
		\end{scope}
	\end{tikzpicture}%
	}
	\def\mycaption{%
		Recursion structure of \proc{RestoreAfterInsert}.
		With probability $p$, we rebalance the entire sublist;
		otherwise, we recurse into one sublist, depending on the rank
		of the newly inserted node (shown in red).\protect\\
		The recursion structure for \proc{RestoreAfterDelete}
		is similar.%
	}
	\ifsiam{%
			\def\mywidth{17em}%
			\begin{minipage}{\mywidth}%
			\vspace*{9ex}
			\caption{\mycaption}%
			\label{fig:restore-after-insert}
			\end{minipage}\hspace*{-1em}%
			\begin{minipage}{\linewidth-\mywidth+1em}%
			\hspace*{-5em}\scalebox{.8}{\usebox\tmpbox}
			\end{minipage}%
	}{}%
	\ifsubmission{%
		\begin{captionbeside}{%
			\mycaption%
		}[l]
			\hspace*{-5em}
			{\scalebox{.8}{\usebox\tmpbox}}%
		\end{captionbeside}
		\label{fig:restore-after-insert}
	}{}%
	\ifarxiv{%
		\plaincenter{\adjustbox{max width=\linewidth}{\usebox\tmpbox}}
		\caption{\mycaption}
		\label{fig:restore-after-insert}
	}{}
\end{figure*}

We briefly sketch the update operations for randomized median-of-$k$ jumplists;
\ifsiam{%
	the \extendedversion%
}{%
	\wref{app:algorithms}%
}
describes them in more detail.
The common theme is that we first modify the jumplist blindly and afterwards 
``repair'' the distribution
by rebuilding one suitably chosen sublist randomly from scratch.
For example upon insertion, the new node has a certain chance to be the target of the
first jump pointer.
We flip a coin to decide whether this should happen; 
if so, we rebuild the entire structure and are done.
Otherwise, we recursively repair a sublist.

\paragraph{Rebalance}
As in~\cite{BronnimannCazalsDurand2003}, we use a procedure \proc{Rebalance}$(\Li)$
that (re)assigns jump pointers from scratch.
It only uses the backbone, existing jump pointers are ignored.
A careful recursive implementation of \proc{Rebalance} rebuilds a 
sublist of $m$ nodes in time $\Theta(m)$.

\paragraph{Insert}
Insertion in jumplists consists of the three phases found in many dictionaries: 
(unsuccessful) search, local insertion, and cleanup.
Unless $x$ is already present,
the search ends at the node with the largest key (strictly) smaller than $x$.
There we insert a new node with key $x$ into the backbone.

It does not have a jump pointer yet, and it is a new potential jump target
for all the nodes whose sublist contains the new node.
Procedure \proc{RestoreAfterInsert} rectifies this as follows.
Let $m$ be the total number of nodes after the insertion, \ie, including the new node.
If $m \le w$, no cleanup is necessary;
if $m=w+1$, we draw the jump pointer for $v_0$ and are done.
Otherwise, we first restore the pointer distribution of $v_0$. 
Due to the insertion of a new node, the sample range now contains an additional node $u$.
($u$ is not necessarily the newly inserted node; 
if the new key is the first or second smallest in \Li, 
$u$ is the former second node of \Li).

If we, conceptually, drew $v_0.\mathit{jump}$ anew,
there are two possibilities: 
either $u$ is part of the sample, namely with probability \smash{$p=\frac{k}{m-2}$},
or $u$ is not part of it.
In the first case, we rebalance all of \Li. 
In the second case,
conditional on the event that $u$ is \emph{not} in the sample,
the current jump pointer of $v_0$ already has the correct distribution:
the median of a random sample not containing~$u$.
We thus rebalance \Li with probability $p$, 
where we draw the jump pointer of $v_0$ conditional on $u$ being part of the sample.
Otherwise we continue recursively
in the uniquely determined sublist that contains the inserted node.
\wref{fig:restore-after-insert} summarizes \proc{RestoreAfterInsert} graphically.

\paragraph{Delete}
We now sketch the procedure \proc{RestoreAfterDelete},
which is similar to \proc{RestoreAfterInsert}.
Let $m$ be the number of nodes after deletion, and let $u$ be the deleted node.
First assume that $u\ne v_0$.
Assume $m > w$, \ie, $v_0$ is a jump node whose sublist contained $u$.
If the sample drawn to choose $v_0.\id{jump}$ did \emph{not} contain $u$,
the deletion of $u$ does not affect $v_0.\id{jump}$,
and we recursively clean up the sublist that formerly contained $u$.
If $u$ was part of the sample, we have to rebalance~\Li;
the probability for that is 
\begin{align*}
	p \wwrel= 
	\begin{cases*}
		1, 		& if $u=v_0.\id{jump}$;\\
		\frac{t}{J_1-1},	& if $u$ was in $\mathcal J_1$; \\
		\frac{t}{J_2-1},	& if $u$ was in $\mathcal J_2$.
	\end{cases*}
\end{align*}
(We define $\frac00\ce1$ in case $t=J_1-1=0$.)
When the deleted node is $u=v_0$, the new header $v_1$ can inherit $v_0$'s jump pointer
and we have the same situation as if $v_1$ had been deleted.
We have to rebalance with probability $p=\frac{t}{J_1-1}$,
otherwise we continue the cleanup in the next-sublist.

\paragraph{Cost Measure}
Insertion and deletion consist of a search and \proc{RestoreAfterInsert/-Delete}.
The latter procedures retrace (a prefix of) the search path to the element
and rebuild at most one sublist using \proc{Rebalance}.
So apart from the search costs (which we analyze separately),
the dominating cost is the number of \emph{``rebalanced elements'':}
the size of the sublists on which \proc{Rebalance} is called.
We will use this as our measure of costs.

\section{Analysis}
\label{sec:analysis}

We now turn to the analysis of the expected behavior of median-of-$k$ jumplists with leaf size $w$.
(The expectation is always over the random choices of the jump pointers.)
We summarize our results in the theorem below.
Its proof is spread over the following subsections.

\begin{theorem}
\label{thm:results}
	Consider randomized median-of-$k$ jumplists with leaf size $w$ on $n$ keys,
	where $k$ and $w$ are fixed constants.
	Abbreviate by $H(k) = \harm{k+1}-\harm{(k+1)/2}$ for $\harm n$
	the harmonic numbers. Then the following holds:
	\begin{thmenumerate}[noitemsep]{thm:results}
	\item \label{thm:results-search}
		The expected number of key comparisons in a \textbf{spine search} is asymptotic to
		\(
				1/H(k) \cdot \ln n 
		\), as $n\to\infty$,
		when each position is equally likely to be requested.
	\item \label{thm:results-insert}
		The expected number of rebalanced elements 
		in the \textbf{cleanup after insertion} is asymptotic to
		\(
				k / H(k)\cdot \ln n
		\), as $n\to\infty$,
		when each of the $n+1$ possible gaps is equally likely.
	\item \label{thm:results-delete}
		The expected number of rebalanced elements 
		in the \textbf{cleanup after deletion} is asymptotic to
		\(
				k/H(k)\cdot \ln n
		\), as \(n\to\infty\),
		when each key is equally likely to be deleted.
	\item \label{thm:results-memory}
		The expected number of additional machine words per key required to 
		store the jumplist is asymptotically at most
		\(1 + \frac2{(w+1)H(k)}\) as $n\to\infty$.
	\end{thmenumerate}
\end{theorem}

\subsection{Search Costs}

Let $P_n$ be the (random) total number of comparisons to search 
all numbers $x\in\{0.5,1.5,\ldots,n+0.5\}$ 
(searching each gap once)
in $\mathcal J_n$ 
the randomized jumplist on $\{1,\ldots,n\}$,
using \proc{SpineSearch}.
The corresponding quantity in BSTs is called external path length,
and we will use this term for $P_n$, as well.
The quotient $P_n/n$ describes the average costs of one call to \proc{SpineSearch}
when all $n+1$ gaps are equally likely to be requested.
$P_n$ is random \wrt to the locations of the jump pointers in $\mathcal J_n$.
To set up a recurrence for~$P_n$, the perspective of random dangling-min BSTs
is most convenient, since \proc{SpineSearch} follows the tree structure.
We describe recurrences here in terms of the distributions of families of random variables.
\begin{align*}
		P_n
	&\rel\eqdist\begin{cases}
		\begin{aligned}
			&(n+1) + (S_n+L_n+1) + S_{J_1} 
			\\&\qquad\bin+P_{J_1} + P'_{J_2},
		\end{aligned}
		& n\ge w,
		\\[2ex]
			\frac{(n+1)(n+2)}{2},
		& n < w,
		\end{cases}
\\[1ex]
		S_n
	&\rel\eqdist
		\begin{cases}
		1 + S_{J_1},\mkern-8mu
		 & n\ge w, \\
		0,
		& n < w,
		\end{cases}
\\[1ex]
		L_n
	&\rel\eqdist
		\begin{cases}
		L_{J_1},\mkern-8mu
		 & n\ge w, \\
		n,
		& n < w,
		\end{cases}
\end{align*}
The terms $P_{J_1}$ and $P'_{J_2}$ on the right-hand side denote members of independent 
copies of the family of random variables $(P_n)_{n\in\N_0}$,
which are also independent of $J_r = \ui{J_r}n$, $r\in\{1,2\}$.
(We omitted the superscripts above for readability.)
Here $J_r = I_r + t$, $r\in\{1,2\}$, $I_1 \eqdist \betaBinomial(n-1-k; t+1,t+1)$
and $J_2 = n-1-k - J_1$.
(We use $n$ here instead of $m$ in~\wref{sec:median-of-k-jumplists-def}; 
hence the slightly different parameters.)

The terms in the expression for $P_n$ are the comparisons with
(1) the root key, 
(2) the dangling min of the root, 
(3) the comparisons done in the left subtree while searching the leftmost gap
(which does not exist in the subtrees any more!),
and (4) the external path lengths of the subtrees.
Two additional quantities are used to express these:
$L_n$ is the number of keys in the leftmost leaf; by definition we have 
$0\le \E{L_n} \le w-1 = \Oh(1)$.
$S_n$~is the number of internal nodes on the ``left \underline spine'' of the tree,
an essential parameter for the linear-search part of \proc{SpineSearch}.
$S_n$ is also the depth of the internal node with the smallest root key 
(ignoring dangling mins).
For ordinary BSTs, $S_n$ is essentially the number of left-to-right minima,
which is a well-understood parameter;
for (fringe-balanced) dangling-min BSTs, such a simple correspondence does not seem to hold.

We point out that the distribution of $P_n$ has a subtle complication, namely
that even conditional on $(J_1,J_2)$,
the quantities
$S_n$, $S_{J_1}$ and $P_{J_1}$ are \emph{not} independent:
all consider the \emph{same} left subtree!
For example, we always have $S_{J_1}=S_n-1$ (for $n\ge w$).
We will only compute the expected value here,
so by linearity, these dependencies can be ignored.

We will derive an asymptotic approximation 
using \wref{thm:DMTwc}, the distributional master theorem (DMT).

\begin{remark}%
For ordinary BSTs, the expectation of above quantities
is known precisely, and some generalizations for fringe-balanced trees 
are possible by solving an Euler differential equation for the generating function.
Unlike there, for dangling-min BSTs the resulting differential equation is \emph{not}
an Euler equation.
The case $t=0$ could be solved since the differential equation 
has order one~\cite{BronnimannCazalsDurand2003},
but there is little hope to obtain a solution 
for the generating function for $t\ge 1$.
\end{remark}

\begin{lemma}
\label{lem:spine-length-asymptotic}
	$\E{S_n} \sim \dfrac{1}{\harm{k}-\harm{t}} \ln n $.
\end{lemma}
\begin{proof}
	We apply \wref{thm:DMTwc} 
	to the distributional recurrence $S_n \eqdist S_{J_1} + 1$.
	It has the form of \wref{eq:DMTwc-distributional-recurrence}
	with (matching the notation of \wref{thm:DMTwc})
	$C_n \DMTvarEq S_n$. We have $s\DMTvarEq 1$ recursive term with size $J_1$
	plus a ``toll term'' $T_n \DMTvarEq 1$. 
	The latter has the asymptotic form $\E{T_n} = 1 \sim 1 \cdot n^0 \lg^0 n$ as $n\to\infty$,
	\ie, $K\DMTvarEq1$, $\alpha\DMTvarEq0$, $\beta\DMTvarEq0$. 
	Moreover, there is no ``coefficient'' in from of the recursive term, so
	$A_1 \DMTvarEq 1$.
	
	We next check the conditions.
	The independence assumptions are trivially fulfilled here, 
	in particular because $T_n$ is a fixed constant.
	We next consider \wref{eq:DMTwc-condition}.
	Recall that $J_1 \eqdist \betaBinomial(n-1-k; t+1,t+1) + t$.
	By \wref{lem:limit-law-beta-binomial} and the remark below it,
	$\ui{Z_1}n = \ui{J_1}n / n$ fulfills
	\begin{align*}
			n \Prob[\big]{\ui{Z_1}n \in (z-\tfrac1n,z] }
		&\wwrel=
			f_{Z_1^*} \bin\pm \Oh(n^{-1}),
	\end{align*}
	for $f_{Z_1^*} : [0,1] \to \R $ with $f_{Z_1^*}(z) = z^{t}(1-z)^{t} / \BetaFun(t+1,t+1)$.
	This function is a polynomial in $z$, so it has bounded derivative (on the compact domain $[0,1]$)
	and is hence Lipschitz continuous (and thus Hölder continuous).
	So \wref{eq:DMTwc-condition} is satisfied with $\delta \DMTvarEq 1$.
	The limiting relative subproblem size $Z_1^*$ has a $\betadist(t+1,t+1)$ distribution.
	
	For the second condition, \wref{eq:DMTwc-condition-coeffs},
	we find that $\E[\big]{ \ui{A_r}n \given \ui{Z_r}n \in (z-\tfrac1n,z] } = 1$
	since $A_1$ is constant.
	So this condition is trivially satisfied with $a_1(z) = 1$ (which is a Hölder-continuous function).
	We have now established that we can apply the DMT to our recurrence.
	
	To obtain the asymptotic approximation for $\E{S_n}$, we consider
	$H=1-\E{(Z_1^*)^0} = 0$, so Case~2 applies: 
	$\E{S_n} \sim \tilde H^{-1}\cdot \E{T_n} \ln n = \tilde H^{-1}\cdot \ln n$
	for the constant $\tilde H = -\sum_{r=1}^s \E{\ln(Z_r^*)}$.
	(Note that this constant only involved the limiting relative subproblem size $Z_r^*$,
	not the relative subproblem size $\ui{Z_1}n$ for a fixed $n$.)
	The expectation in $\tilde H$ is exactly the first part of \wref{lem:E-ln-D}, 
	so we find $\tilde H = \harm{k}-\harm{t}$.
	Now the claim follows by inserting above.
\end{proof}

\begin{remark}[Spine lengths]
\wref{lem:spine-length-asymptotic} implies that the expected left spine 
of the root is logarithmic~-- as one might expect in a random BST;
indeed, the expected left spine lengths of the root in a random BST and a dangling-min BST
differ only in lower order terms.
Note that the former is exactly $\harm{n}$ and the proof is elementary:
The left spine length in a BST is the number of left-to-right minima in the insertion order.
For dangling-min BSTs, no such simple argument is available.
\end{remark}
With these preparations, we can prove the main statement about search costs.

\begin{proof}[\wref{thm:results-search}]
	We again use the distributional master theorem (DMT);
	this time on the recurrence
	$P_n
		\rel\eqdist (n+1) + (S_n+L_n+1) + S_{J_1} 
				+P_{J_1} + P'_{J_2}$.
	The recurrence is more involved than the one for $S_n$
	that we just solved,
	but the distribution of subproblem sizes are the same,
	and we again have no coefficient in front of the recursive terms.
	Therefore, a large part of the argument can be copied from the proof
	of \wref{lem:spine-length-asymptotic}.
	
	We here have
	$C_n \DMTvarEq P_n$, there are  $s\DMTvarEq2$ recursive terms and 
	$T_n \DMTvarEq (n+1) + (S_n+L_n+1) + S_{J_1}$. 
	By \wref{lem:spine-length-asymptotic},
	all but the first summand in $\E{T_n}$ are actually in $\Oh(\log n)$, so from
	the initially complicated toll function, only $\E{T_n} \sim n$ remains in the leading term 
	as $n\to\infty$.
	We thus have $K\DMTvarEq 1$, $\alpha\DMTvarEq1$, $\beta\DMTvarEq 0$.
	
	The coefficients $A_r \DMTvarEq 1$ for $r\in\{1,2\}$, so \wref{eq:DMTwc-condition-coeffs}
	again holds trivially with $a_r(z)=1$.
	As in the proof of \wref{lem:spine-length-asymptotic}, 
	$Z_1^* \eqdist Z_2^* \eqdist \betadist(t+1,t+1)$ holds and condition 
	\wref{eq:DMTwc-condition} holds with the same $f_{Z_1^*}$.
	We find again $H=0$ (since $Z_1^* + Z_2^* = 1$), so Case~2 applies. 
	The constant $\tilde H$ this time involves the second part of \wref{lem:E-ln-D}:
	$\tilde H = -\sum_{r=1}^s \E{D_r \ln(D_r)} = \harm{k+1}-\harm{t+1}$.
	So we have 
	\(
			\E{P_n}
		\wwrel\sim 
			\frac{1}{\harm{k+1}-\harm{t+1}} n \ln n 
	\)
	and dividing by $n+1$ yields the claim.
\end{proof}

\subsection{Insertion Costs}
\label{sec:analysis-insert}

The steps taken by \proc{RestoreAfterInsert} depend on the position of the newly inserted element;
we denote here by $R$ the rank of the gap the new element is inserted into.
When the current sublist has $m$ nodes, we have $R\in[0..m]$.
Similar as for searches, we consider the average costs of insertion when all possible gaps are
equally likely to be requested.

Unlike for searches, the distribution of $R'$ in subproblems is \emph{not} uniform
even when $R$ is:
a close inspection of \proc{RestoreAfterInsert} reveals that 
(a) recursive calls in the jump-sublist always have $R'\ge 1$, and
(b) $R=0$ and $R=1$ yield $R'=0$ in the recursive call in the next-sublist;
in fact, once $R=0$ holds, we get this rank in all later recursive calls.
We can therefore handle this by splitting the cases $R=0$ and $R\ge 1$;
Also note that for the topmost call to \proc{RestoreAfterInsert},
$R=0$ is not possible, since no insertion before the header with dummy-key $-\infty$ is
possible.
This means that initially $R\eqdist \uniform[1..m]$ holds.
Recall that a jumplist on $m$ nodes stores only $n=m-1$ keys, so that there are only $n+1=m$ possible
gaps.
We obtain the following distributional recurrence for $B_m^{\ins}$, the random
number of re\underline balanced elements during insertion into the $R$th gap 
in a randomized median-of-$k$ jumplist with $m$ nodes.
(Note that \ifsiam{}{unlike in the pseudocode,} 
$m$ is here the number of nodes in the jumplist \emph{before} the insertion.)
\begin{align*}
		B_m^{\ins}
	&\rel\eqdist
		\begin{dcases}
			\begin{aligned}
				&F \cdot (m+1)
			\\[-.75ex]&\quad
				+
				(1-F) \Bigl( 
					 \indicator{R = 1} B_{J_1}^{\ins0} 
			\\[-.75ex]&\qquad
					+\indicator{2\le R \le J_1+1} B_{J_1}^{\ins} 
			\\[-.25ex]&\qquad
					+\indicator{R \ge J_1+2} B_{J_2}^{\ins} 
				\Bigr),
			\end{aligned}
		& 	m > w,
		\\[1ex]
			[m=w]\cdot(m+1),
		&	m\le w,
		\end{dcases}
\\[1ex]
		B_m^{\ins0}
	&\rel\eqdist
		\begin{dcases}
			\begin{aligned}
				&F \cdot (m+1) \bin+ (1-F) B_{J_1}^{\ins0},
			\end{aligned}
		& 	m > w,
		\\[1ex]
			[m=w]\cdot(m+1),
		&	m\le w,
		\end{dcases}
\\[1ex]
	&\mkern-20mu\text{where}\quad
		R
	\rel\eqdist 
		\uniform[1..m],\quad\;
		F
	\rel\eqdist 
		\bernoulli\Bigl(\frac k{m-1}\Bigr),
\end{align*}
All $B_m$ terms on the right-hand side denote independent copies of the family of
random variables and $R$ and $F$ are independent of $B_m$ and $(J_1,J_2)$.
Here $J_r = I_r + t+1$, $r\in\{1,2\}$, $J_1 \eqdist \betaBinomial(m-2-k; t+1,t+1)$
and $J_2 = m-2-k - J_1$ (as in~\wref{sec:median-of-k-jumplists-def}).

\begin{lemma}
\label{lem:B-n-ins0}
	$\E{B_m^{\ins0}} \sim \dfrac{k}{\harm{k} - \harm{t}}\ln m$.
\end{lemma}
\begin{proof}
We use once more the distributional master theorem.
As before, $Z_1^* \eqdist \betadist(t+1,t+1)$ and the condition \wref{eq:DMTwc-condition} 
is satisfied by \wref{lem:limit-law-beta-binomial}.
We have $\E{T_n} \DMTvarEq \E{F(n+1)} \sim k = \Theta(1)$.
Unlike before, we here have a non-constant coefficient $\ui{A_1}n \DMTvarEq 1-F$ in front of 
the recursive term, but since
$\E{1-F} = 1 \pm \Oh(n^{-1})$, \wref{eq:DMTwc-condition-coeffs} is again fulfilled with $a_1(z) = 1$.
As in the proof of \wref{lem:spine-length-asymptotic},
we find $H = 0$ (Case~2) and with the claim follows from
$\tilde H = - \E{\ln D_1} = \harm{k} - \harm{t}$ (\wref{lem:E-ln-D}).
\end{proof}

\begin{proof}[\wref{thm:results-insert}]
Towards applying the DMT on $C_n\DMTvarEq B_n^{\ins}$, we compute
\begin{align*}
		\E{T_n} 
	&\wwrel\DMTvarEq
		\E[\Big]{F(n+1) + (1-F)\indicator{R=1}B_n^{\ins0}}
\\	&\wwrel=
		\frac {k(n+1)}{n-1} \bin+ \frac{n-1-k}{n-1} \cdot \frac1n \cdot \E{B_n^{\ins0}}
\\	&\wwrel{\eqwithref[r]{lem:B-n-ins0}}
		k \bin\pm \Oh(n^{-1}\log n).
\end{align*}
As usual, we have $Z_r^* \eqdist \betadist(t+1,t+1)$, $r\in\{1,2\}$,
and \wref{eq:DMTwc-condition} is fulfilled by
\wref{lem:limit-law-beta-binomial}.
For the coefficients of the recursive terms holds
\begin{align*}
		&\E*{\ui{A_1}n \given \ui{Z_1}n \in (z-\tfrac1n,n]}
\\	&\wwrel\DMTvarEq
		\Prob[\Big]{2\le R\le J_1+1 \given \ui{Z_1}n \in (z-\tfrac1n,n]}
\\	&\wwrel=
		\Prob[\Big]{\tfrac{J_1}{n} \given \ui{Z_2}n \in (z-\tfrac1n,n]}
\\	&\wwrel=
		\Prob[\big]{\ui{Z_1}n \given \ui{Z_1}n \in (z-\tfrac1n,n]}
\\	&\wwrel=
		z \bin\pm \Oh(n^{-1})
		,
\shortintertext{and similarly}
		&\E[\big]{\ui{A_2}n \given \ui{Z_2}n \in (z-\tfrac1n,n]}
\\	&\wwrel\DMTvarEq
		\Prob[\big]{R\ge J_1+2 \given \ui{Z_2}n \in (z-\tfrac1n,n]}
\\	&\wwrel=
		z \bin\pm \Oh(n^{-1})
		,
\end{align*}
so that \wref{eq:DMTwc-condition-coeffs} holds with $a_1(z) = a_2(z) = z$,
and we can apply the DMT.
Since $H=1-\sum_{r=1}^2 \E[\big]{(Z_r^*)^0 \*a_r(Z_r^*)}=1-\sum_{r=1}^2 \E{Z_r^*}=0$, 
we again have Case~2 and find
$\tilde H = -\sum_{r=1}^2 \E{D_r \ln D_r} = \harm{k+1}-\harm{t+1}$
with \wref{lem:E-ln-D}.
This proves the claim.
\end{proof}

\subsection{Deletion Costs}
\label{sec:analysis-delete}

As for insertion, we analyze the size of the sublist $B_m^{\del}$ 
that is rebuilt using \proc{Rebalance} 
when the rank of the deleted element is chosen uniformly.
Initially, we have $R\eqdist \uniform[2..m]$ since
the dummy key $-\infty$ in the header cannot be deleted.
In recursive calls, also $R=1$ is possible, 
and we remain in this case for good whenever we enter it once.
We can thus characterize the deletion costs using the two quantities $B_m^{\del}$ and
$B_m^{\del1}$.
As for insertion, $m$ is the ``old'' size of the jumplist, \ie, 
the number of nodes \emph{before} the deletion.
\begin{align*}
		B_m^{\del}
	&\rel\eqdist
		\begin{dcases}
			\begin{aligned}
				&F \cdot (m-1)
			\\[-.75ex]&\quad
				+
				(1-F) \Bigl( 
					 \indicator{R = 2} B_{J_1}^{\del1} 
			\\[-.75ex]&\qquad
					+\indicator{3\le R \le J_1+1} B_{J_1}^{\del} 
			\\[-.25ex]&\qquad
					+\indicator{R \ge J_1+3} B_{J_2}^{\del} 
				\Bigr),
			\end{aligned}
		& 	m > w,
		\\[1ex]
			[m=w]\cdot1,
		&	m\le w,
		\end{dcases}
\\[1ex]
	&\mkern-20mu\text{where}\;\;
		R
	\rel\eqdist 
		\uniform[2..m],\quad
		\text{and cond.\ on $(R,J_1,J_2)$}
\\		&\phantom{\mkern-20mu\text{where}\;\;} F
	\rel\eqdist 
		\begin{cases}
			\bernoulli\bigl(\frac t{J_1-1}\bigr), & R \le J_1 + 1;\\
			1, & R = J_1 + 2;\\
			\bernoulli\bigl(\frac t{J_2-1}\bigr), & R \ge J_1 + 3,
		\end{cases}
\\[1.5ex]
		B_m^{\del1}\!
	&\rel\eqdist
		\begin{dcases}
			F_1 \cdot (m-1) + (1-F_1) B_{J_1}^{\del1},
		& 	m > w,
		\\[1ex]
			[m=w]\cdot1,
		&	m\le w,
		\end{dcases}
\\
	&\mkern-20mu\text{where cond.\ on $J_1$}\;\;
		F_1
	\rel\eqdist 
		\bernoulli\Bigl(\frac t{J_1-1}\Bigr).
\end{align*}
As before, the $B_m$ terms on the right are independent copies of the family of
random variables and $R$ and $F$/$F_1$ are independent of $B_m$ and $(J_1,J_2)$.
We have $J_r = I_r + t+1$, $r\in\{1,2\}$, $J_1 \eqdist \betaBinomial(m-2-k; t+1,t+1)$
and $J_2 = m-2-k - J_1$.
The (asymptotic) solution of these recurrences is similar to the case of insertion,
but a few more complications arise.

\begin{lemma}
\label{lem:B-n-del1}
	For $t=0$ we have $\E{B_m^{\del1}} \le 1$.\needhspace{8em}
	If $t\ge 1$, $\E{B_m^{\del1}} \sim \dfrac{k}{\harm{k}-\harm{t}} \ln m$.
\end{lemma}
\begin{proof}
	For $t=0$, we have $F_1=0$ (almost surely) in each iteration, so the recurrence
	collapses to its initial condition, which is at most $1$.
	In the following, we now consider $t\ge1$.
	The proof will ultimately use the DMT on $C_n \DMTvarEq B_n^{\del1}$, 
	but we need a few preliminary results 
	to compute the toll function $\E{T_n} \DMTvarEq \E{F_1(n-1)}$.
	We write the $a=b\pm d$ to mean $b-d \le a \le b+d$ here and throughout.
	With that notation, we give the following elementary approximation:
	\begin{align}
	\label{eq:t-over-m-plus-t}
		\forall t \in\N_{\ge1} \:\forall n\ge0 \rel: 
		\smash{\frac{t}{n+t}} \rel= t n^{\underline{-1}} \bin\pm t(t-1) n^{\underline{-2}}
		.
	\end{align}
	Now, we compute the expectation of $F_1$ conditional on $I_1 = J_1 - t - 1$.
	\begin{align*}
			\E{F_1\given I_1}
		&\wrel=
			\frac t{J_1-1}
		\wrel=
			\frac{t}{I_1+t}
	\\	&\wrel{\eqwithref{eq:t-over-m-plus-t}}
			t \cdot I_1^{\underline{-1}} \wbin\pm t(t-1)\cdot I_1^{\underline{-2}}
			.
	\end{align*}
	Next, we use the stochastic representation of beta-binomials
	(recall \wref{sec:preliminaries});
	we take expectations over $I_1\eqdist\binomial(\eta,D_1)$ with $\eta=m-2-k$, 
	but conditional on $D_1$. We write $D_2=1-D_1$. Then it holds that
	\begin{align*}
			\quad&\mkern-25mu
			\E{F_1 \given D_1}
	\\[-.5ex]	&\wrel{\eqwithref[r]{lem:binomial-negative-factorial-moments}}
			\frac t{\eta+1} D_1^{-1} (1-D_2^{\eta+1}) \bin\pm t(t-1) D_1^{-2} \eta ^{\underline{-2}}
			.
	\end{align*}
	Finally, we also compute the expectation \wrt $D_1 \eqdist \betadist(t+1,t+1)$;
	note that for $t\ge 2$, $\E{D_1^{-2}}$ exists and has a finite value (independent of $n$);
	whereas for $t=1$, the error term is zero.
	So we find in both cases with \wref{lem:powers-to-parameters}:
	\begin{align*}
			\E{F_1}
		&\rel=
			\frac t{\eta+1} \E{D_1^{-1}} 
			-\frac t{\eta+1} \E{D_1^{-1}D_2^{\eta+1}} 
			\bin\pm \Oh(\eta^{-2})
	\\	&\rel=
			\frac t{\eta+1} \frac{k}{t}
			-\frac t{\eta+1} \frac{(t+1)^{\overline{\eta+1}}}{t(k+1)^{\overline{\eta}}}
			\bin\pm \Oh(\eta^{-2})
	\\	&\rel=
			\frac k{\eta+1}
			-\frac {(t+1)(t+2)}{(\eta+1)(\eta+2)} 
				\underbrace{ \frac{(t+3)^{\overline{\eta-1}}}{(k+2)^{\overline{\eta-1}}}}_{<1}
			\bin\pm \Oh(\eta^{-2})
	\\[-3ex]	&\rel=
			\frac k{\eta+1}
			\bin\pm \Oh(\eta^{-2})
			.
	\numberthis\label{eq:E-F-del}
	\end{align*}
	With this we finally get $\E{T_n} \DMTvarEq \E{F_1 (n-1)} = k \pm \Oh(n^{-1})$.
	$Z_1^* \eqdist \betadist(t+1,t+1)$ and fulfills \wref{eq:DMTwc-condition}.
	For \wref{eq:DMTwc-condition-coeffs}, we compute
	\begin{align*}
			&\E[\big]{\ui{A_1}n\given \ui{Z_1}n \in (z-\tfrac1n,z]} 
	\\	&\wrel= 
			\E[\big]{1-F_1\given \ui{Z_1}n \in (z-\tfrac1n,z]}
	\\	&\wrel=
			1 \pm \Oh(n^{-1}).
	\end{align*}
	So the DMT applies;
	we have $H=0$, \ie, Case~2. The claim follows with  
	$\tilde H = -\E{\ln Z_1^*} = \harm{k}-\harm{t}$.~
\end{proof}

\begin{proof}[\wref{thm:results-delete}]
We start with computing the conditional expectation of $F$, the coin flip indicator.
\begin{align*}
		\E{F\given J_1}
	&\rel=
		 \frac{J_1}{n-1}  \frac{t}{J_1-1}
		+\frac1{n-1}  1
		+\frac{J_2-1}{n-1} \frac{t}{J_2-1}
\\	&\rel=
		\frac{2t+1}{n-1} \bin+ \frac1{n-1} \cdot \frac t{J_1-1}
		.
\shortintertext{Hence}
		\E{F}
	&\rel{\eqwithref{eq:E-F-del}}
		\frac{2t+1}{n-1} \bin+ \frac1{n-1} \cdot \frac{k}{\eta+1} \pm \Oh(n^{-3})
\\	&\rel=
		\frac{k}{n-1} \bin\pm \Oh(n^{-2})
		.
\end{align*}
Towards applying the DMT on $C_n\DMTvarEq B_n^{\del}$, we compute
\begin{align*}
		\E{T_n} 
	&\wrel\DMTvarEq
		\E[\Big]{F(n-1) + (1-F)\indicator{R=2}B_n^{\del1}}
\\[-.5ex]	&\wrel{\eqwithref[r]{lem:B-n-del1}}
		k \wbin\pm \Oh(n^{-1} \log n)
		.
\end{align*}
We have $Z_r^* \eqdist\betadist(t+1,t+1)$ and 
\wref{eq:DMTwc-condition} is fulfilled.
Similarly as in \wref{sec:analysis-insert},
we find that \wref{eq:DMTwc-condition-coeffs} holds with
$a_1(z) = a_2(z) = z$.
Once more we have $H=0$ and Case~2 applies,
and the claim follows
with $\tilde H = -\sum_{r=1}^2 \E{D_r \ln D_r} = \harm{k+1}-\harm{t+1}$.
\end{proof}

\subsection{Memory Requirements}

We assume that
a pointer requires one word of storage, and so does an integer that can take values in $[0..n+1]$.
We do not count memory to store the keys since any (general-purpose) 
data structure has to store them.
This means that a plain node requires $1$ word of (additional) storage, 
and a jump node needs $3$ additional words (two pointers and one integer).
Let $A_n$ denote the (random) number of jump nodes, excluding the dummy header, 
of a random median-of-$k$ jumplist with leaf size $w$ on $n$ keys, 
then its additional memory requirement is
$3 (A_n+1) + 1(n-A_n)$.
It remains to show that $A_n$ is asymptotically at most
\( 1/ \bigl((w+1)(\harm{k+1}-\harm{t+1})\bigr) n \).

$A_n$ counts the internal nodes in a random
fringe-balanced dangling-min BST over $n$ keys;
a distributional recurrence is thus easy to set up:
\begin{align*}
		A_n
	&\wrel\eqdist
		\begin{cases}
			1 + A_{J_1} + A_{J_2}, 
		&	n > w-1,
		\\
			0, 
		&	n \le w-1.
		\end{cases}
\end{align*}
Here again $J_r = I_r + t$, $r\in\{1,2\}$, $J_1 \eqdist \betaBinomial(n-1-k; t+1,t+1)$
and $J_2 = n-1-k - J_1$.
For $A_n$, the DMT only gives us $\E{A_n} = \Oh(n)$ (Case~3).
It is easy to see that $\E{A_n}$ is also $\Omega(n)$,
but a precise leading-term seems very hard to obtain.

\begin{proof}[\wref{thm:results-memory}]
	The recurrence for $A_n$ is very similar to that for the number of partitioning steps
	in median-of-$k$ Quicksort with Insertionsort threshold $w-1$;
	the only difference is that we there have $I_1\eqdist I_2 \eqdist \betaBinomial(n-k;t+1,t+1)$,
	\ie, with $n-k$ instead of $n-k-1$.
	By monotonicity, $\E{A_n}$ is at most the number of partitioning steps in Quicksort
	since also the subproblems sizes are smaller.
	The number of partitioning steps in median-of-$k$ Quicksort with Insertionsort threshold $M$
	is $1/\bigl((M+2)(\harm{k+1}-\harm{t+1})\bigr) n \pm \Oh(1)$, see, \eg,
	\cite[p.\,327]{Hennequin1989}.
	Setting $M=w-1$ yields the claim.
\end{proof}

\section{Conclusion}
\label{sec:conclusion}

In this article, we presented median-of-$k$ jumplists and analyzed their 
efficiency in terms of the expected number of comparisons (for searches) 
and rebalanced elements (for updates). 
The precise analysis of insertion and deletion costs is also novel for
the original version of jumplists ($k=1$).

Our analysis shows that a search profits from sampling;
in particular going from $k=1$ to $k=3$ entails significant savings:
$\frac{12}7 \ln n \approx 1.714\ln n$ instead of $2 \ln n$ comparisons
on average.
As for median-of-$k$ Quicksort, we see diminishing returns for much larger~$k$.
For jumplists, also the cleanup after insertions and deletions gets more expensive;
the effort grows linearly with $k$. Very large $k$ will thus be harmful.

The efficiency of insertion and deletion depends on
both the time for search and the time for cleanup, 
so it is natural to ask for optimal $k$.
Since the cost units are rather different 
(comparisons vs.\ rebalanced elements) we need a weighing factor.
Depending on the relative weight $\xi\in[0,1]$ of comparisons,
we can compute optimal~$k$, see \wref{fig:optimal-k}.
In the realistic range, we should try $k=1$, $3$, or $5$,
unless we do many more searches than~updates.

\begin{figure}
	\plaincenter{\adjustbox{max width=.9\linewidth}{\includegraphics[scale=.55]{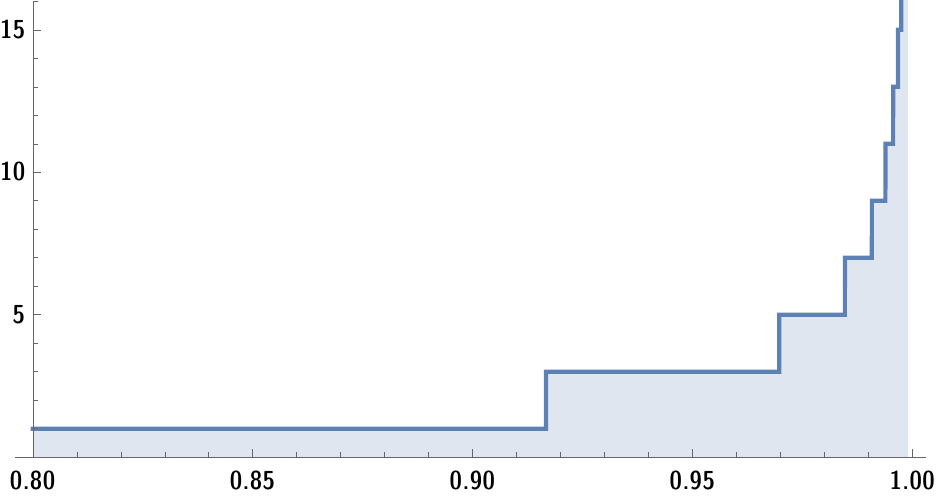}}}%
	\caption{%
		The $k$ that minimizes the leading-term coefficient of
		total costs of insertion/deletion, 
		if one comparison costs $\xi\in[0,1]$ and each rebalanced element costs $1-\xi$,
		\ie, $\arg \min_{k} \xi \cdot \frac1{H(t)} + (1-\xi) \cdot \frac k{H(t)}$
		as a function of $\xi$.%
	}%
	\label{fig:optimal-k}
\end{figure}

We conducted a small running time study based on a proof-of-concept implementation~\cite{codeWild2016} 
in Java
that confirms our analytical findings:
Sampling leads to some savings for searches, 
but slows down insertions and deletions significantly. 
Comparing running times with that of Java's \texttt{TreeMap} (a red-black tree implementation) 
shows that our data structure is only partially competitive: 
for iterating over all elements, jumplists are about 50\% faster, 
but searches are between 20\% and 100\% slower (depending on the choice for $w$) 
and for insertions/deletions \texttt{TreeMap}s are 5 to 10 times faster. 
However, \texttt{TreeMap}s use 4 additional words per key 
(without even storing subtree sizes needed for efficient rank-based access),
whereas our jumplists never need more than $\sim 2.\overline 3$ additional words per key
and less than $1.04$ with $w\approx100$.
For $n=10^6$ keys, $w\approx 100$ did not affect searches much ($+25\%$) 
but actually sped up insertions and deletions (roughly by a factor of~2!).

\subsection{Future Work}
Some interesting questions are left open. 
What is the optimal choice for $w$? 
Answering this question requires second-order terms of search, insertion and deletion costs; 
due to the underlying mathematical challenges it is unlikely that those can be computed exactly, 
but an upper bound using analysis results on Quicksort should be possible. 
Other future directions are 
the analysis of branch misses, in particular in the context of an asymmetric sampling strategy, 
and the design of a ``bulk insert'' algorithm that is faster than inserting elements subsequently,
one at a time. 

On modern computers the cache performance of data structures is important for their 
running time efficiency. 
Here, a larger fanout of nodes is beneficial since it reduces the expected number of I/Os. 
For jumplists this can be achieved by using more than one jump pointer in each node. 
The case of two jump pointers per node has been worked out in detail~\cite{Neumann2015},
but the general scheme invites further investigation.

\ifsiam{}{

\input{jumplists-appendix}

}

\bibliography{jumplist-references.bib}

\ifdraft{
	\clearpage
	\part*{Notes-to-self}
	\printnotestoself
}{}

\end{document}

%% file: jumplists-preamble.tex
\RequirePackage{etex}
\RequirePackage{fixltx2e}

\makeatletter

\usepackage{lmodern}
\usepackage[utf8]{inputenc}
\usepackage[T1]{fontenc}
\usepackage{babel}
\input{ushyphex.tex} 

\usepackage{amsmath,amssymb,amsfonts,mathtools,colonequals}
\mathtoolsset{mathic=true}
\usepackage{calc,relsize,xifthen,xcolor,chngcntr,xspace,adjustbox,pbox}
\usepackage{array,multirow,multicol,booktabs,rotating}
\usepackage{enumitem}
\usepackage{placeins,needspace}
\usepackage{mleftright}\mleftright 
\usepackage{url}

\usepackage{xfrac,dsfont,bm} 
\usepackage{stmaryrd} 

\newcommand{\ifarxiv}[2]{\ifthenelse{\equal{\version}{arxiv}}{#1}{#2}}
\newcommand{\ifdraft}[2]{\ifthenelse{\equal{\draftmode}{true}}{#1}{#2}}
\newcommand{\ifsiam}[2]{\ifthenelse{\equal{\version}{proceedings-siam}}{#1}{#2}}
\newcommand{\ifsubmission}[2]{\ifthenelse{\equal{\version}{submission}}{#1}{#2}}

\ifdraft{}{%
	\usepackage{microtype}
}

\usepackage{wref}
\usepackage{needhspace}

\ifsiam{}{
	\setlength\parindent{1.5em}
	\usepackage[headsepline]{scrlayer-scrpage}
	\pagestyle{scrheadings}
	\clearscrheadfoot
	\AtBeginDocument{%
		\manualmark%
		\markleft{\mytitle}%
		\automark*[section]{}%
	}
	\ohead{\pagemark}
	\ihead{\headmark}
	\addtokomafont{caption}{\sffamily\small}
	\addtokomafont{captionlabel}{\sffamily\textbf}
	\setcapmargin{2em}
}
\ifsubmission{
	\setcapmargin{.75em}
	\setcapindent{0pt}
}{}
\ifarxiv{
	\setcapmargin{.75em}
	\setcapindent{0pt}
}{}

\usepackage{tikz}
\usetikzlibrary{positioning,external,calc}
\usetikzlibrary{decorations.pathreplacing,shapes.multipart,shapes.geometric,shapes.misc}
\usetikzlibrary{fit}
\usepackage{ifluatex}
\ifluatex
	\usetikzlibrary{graphdrawing}
	\usegdlibrary{trees,layered}
\else \fi

\tikzsetexternalprefix{pics/externalized/}
\tikzset{
	external/system call={%
		lualatex \tikzexternalcheckshellescape -halt-on-error %
			-interaction=batchmode -jobname "\image" "\texsource"%
	},
}
\tikzset{external/export=false} 

\usepackage[outline]{contour}

\usepackage{pgfplots}
\pgfplotsset{compat=1.12}

\pgfdeclarelayer{foreground}
\pgfdeclarelayer{background}
\pgfsetlayers{background,main,foreground}

\tikzset{%
	every picture/.style={
		>=stealth,
		font={\sffamily\small},
	},
}

\pgfkeys{%
	/tikz/on layer/.code={
		\pgfonlayer{#1}\begingroup
		\aftergroup\endpgfonlayer
		\aftergroup\endgroup
	},
	/tikz/node on layer/.code={
		\pgfonlayer{#1}\begingroup
		\expandafter\def\expandafter\tikz@node@finish\expandafter{\expandafter\endgroup\expandafter\endpgfonlayer\tikz@node@finish}%
	},
}

\newcommand\currentcoordinate{\the\tikz@lastxsaved,\the\tikz@lastysaved}
\newcommand\currentx{\the\tikz@lastxsaved}
\newcommand\currenty{\the\tikz@lastysaved}

\tikzset{%
	highlightbox/.style={
		draw=softhighlight,
		very thick,
		rounded corners=3pt,
	},
}

\usepackage{pgfplots}
\pgfplotsset{compat=1.12}

\ifsiam{
	\usepackage[margin=.75em,font={small,rm},labelfont={bf}]{caption}
	\DeclareCaptionStyle{ruled}{%
		format=plain,%
		labelfont=bf,%
		labelsep=period,%
		strut=on,%
		margin=.3em,%
		font={small,rm}%
	}
}{}

\usepackage{clrscode3e}

\newcommand*\EndIf{\End\li\kw{end if} }
\newcommand*\EndWhile{\End\li\kw{end while} }
\newcommand*\New{\kw{new}\ }

\usepackage{float}
\floatstyle{ruled}
\newfloat{algorithm}{tbp}{loa}
\floatname{algorithm}{Algorithm}

\hypersetup{
	final,
	unicode=true, 
	bookmarks=true,
	bookmarksnumbered=true,
	bookmarksdepth=2,
	bookmarksopen=true,
	breaklinks=true,
	hidelinks,
}

\usepackage{attachfile2}

\usepackage[amsmath,hyperref,thmmarks]{ntheorem}

\theorembodyfont{\slshape}
\theoremseparator{:}
\newtheoremstyle{proofstyle}%
  {\item[\theorem@headerfont\hskip\labelsep ##1\theorem@separator]}%
  {\item[\theorem@headerfont\hskip\labelsep ##1 of ##3\theorem@separator]}

\theoremstyle{break}
\ifsubmission{
	\theorempreskip{.59\baselineskip}
	\theorempostskip{.5\baselineskip}
}{
	\theorempreskip{2\topsep}
}
\ifsiam{
	\renewtheorem{theorem}{Theorem}[section]
}{
	\newtheorem{theorem}{Theorem}[section]
}

\theoremstyle{plain}
\theorempreskip{\topsep}

\ifsiam{
	\renewtheorem{lemma}[theorem]{Lemma}
	\renewtheorem{corollary}[theorem]{Corollary}
	\renewtheorem{proposition}[theorem]{Proposition}
}{
	\newtheorem{lemma}[theorem]{Lemma}

}

\theoremstyle{plain}
\theorembodyfont{\upshape}

\newtheorem{remark}{Remark}[section]

\theoremsymbol{\raisebox{-.25ex}{$\Box$}}
\qedsymbol{\raisebox{-.25ex}{$\Box$}}

\theoremstyle{proofstyle}
\ifsiam{
	\renewtheorem{proof}{Proof}
}{
	\newtheorem{proof}{Proof}
}

\newenvironment{thmenumerate}[2][]{%
	\begin{enumerate}[
		label={\makebox[\widthof{(a)}][c]{\textup{(\alph*)}}},
		ref={\ref{#2}\kern.1em--\kern.1em(\alph*)},
		itemsep=0pt,
		#1
	]%
}{%
	\end{enumerate}%
}

\vfuzz \hfuzz	

\ifarxiv{
	\raggedbottom
}{}

\allowdisplaybreaks[3]

\pgfplotsset{
	legend cell align=left,
	every axis legend/.append style={%
		fill=backgroundshading!50!white,
		rounded corners=1pt,
	},
	xlabel near ticks,
	ylabel near ticks,
}
\pgfplotsset{
	every y tick scale label/.style={
		at={(-.03,1)},yshift=-1ex,anchor=south east,inner sep=0pt
	}
}

\def\mathematicaPlotAddPlot{
	\@ifstar\@mathematicaPlotAddPlot\@@mathematicaPlotAddPlot%
}
\newcommand\@@mathematicaPlotAddPlot[4][]{
	\addplot+[#1] table[x=#2,y=#3] {\tablefile} ;
	\ifthenelse{\equal{#4}{}}{}{\addlegendentry{#4}}
}
\newcommand\@mathematicaPlotAddPlot[4][]{
	\addplot[#1] table[x=#2,y=#3] {\tablefile} ;
	\ifthenelse{\equal{#4}{}}{}{\addlegendentry{#4}}
}

\pgfplotscreateplotcyclelist{coloredlineplots}{%
	thick,headingscolor\\%
	thick,black\\%
	ultra thick,headingscolor,densely dotted\\%
	very thick,black,densely dashed\\%
	semithick,headingscolor,dashdotted\\%
	thin,black,loosely dashed\\%
}

\pgfplotscreateplotcyclelist{coloreddiscreteplots}{%
	mark=*,headingscolor,every mark/.append style={fill=headingscolor!50}%
\\%
	mark=square*,every mark/.append style={fill=white}%
\\%
	densely dashed,mark=square*,headingscolor,%
		every mark/.append style={solid,fill=red!10}%
\\%
	densely dashed,mark=*,%
		every mark/.append style={solid,fill=black!50}%
\\%
	mark=diamond*,every mark/.append style={%
		semithick,solid,fill=headingscolor!30!black},%
	headingscolor,very thick,densely dotted,%
\\%
}

\newdimen\makeboxdimen
\newcommand\makeboxlike[3][l]{%
\setbox0=\hbox{#2}%
\global\makeboxdimen=\wd0%
\setbox1=\hbox{\makebox[\makeboxdimen][#1]{%
\makebox[0pt][#1]{#3}%
}}%
\ht1=\ht0%
\dp1=\dp0%
\box1%
}

\newcommand\plaincenter[1]{%
	\mbox{}\hfill#1\hfill\mbox{}%
}

\newcounter{inlineenum}
\newenvironment{inlineenumerate}{%
	\setcounter{inlineenum}{0}%
	\def\item{%
		\stepcounter{inlineenum}%
		(\arabic{inlineenum})
	}%
}{%
}

\newcommand*\ie{\mbox{i.\hspace{.2ex}e.}}
\newcommand*\eg{\mbox{e.\hspace{.2ex}g.}}

\newcommand*\wrt{\mbox{w.\hspace{.2ex}r.\hspace{.2ex}t.}\xspace}

\newcommand\parenthesisclause[1]{%
	\unskip%
	\nobreak\hspace{.2ex}\mbox{---}\hspace{.2ex}%
	#1%
	\nobreak\hspace{.2ex}\mbox{---}\hspace{.2ex}%
	\ignorespaces%
}
\newcommand*\parenthesissign{%
	\unskip%
		\hspace{.2ex}---\hspace{.2ex}%
	\ignorespaces%
}

\newcommand\literalquote[3]{%
	{%
		\def\ellipsis{[\,\dots]\xspace}%
		\def\comment##1{[##1]}%
		\def\cite##1{##1}%
		\edef\citekey{#2}%
		{\itshape``\ignorespaces#1\unskip''}
		(%
			\ifthenelse{\equal{#3}{}}{%
			\citet{\citekey}%
			}{%
			\citet{\citekey}%
			, p.\,#3%
			}%
		)%
	}%
	\xspace%
}
\newcommand\literalquotegerman[3]{%
	{%
		\def\ellipsis{[\,\dots]\xspace}%
		\def\comment##1{[##1]}%
		\def\cite##1{##1}%
		\edef\citekey{#2}%
		{\selectlanguage{ngerman}%
		\itshape,,\ignorespaces#1\unskip``}
		(%
			\ifthenelse{\equal{#3}{}}{%
			\citet{\citekey}%
			}{%
			\citet{\citekey}%
			, p.\,#3%
			}%
		)%
	}%
	\xspace%
}
\newcommand\literalquotep[3]{%
	{%
		\def\ellipsis{[\,\dots]\xspace}%
		\def\comment##1{[##1]}%
		\def\cite##1{##1}%
		\edef\citekey{#2}%
		{\itshape``\ignorespaces#1\unskip''}
		(%
			\ifthenelse{\equal{#3}{}}{%
			\citep{\citekey}%
			}{%
			\citep{\citekey}%
			, p.\,#3%
			}%
		)%
	}%
	\xspace%
}

\newcommand{\ESymbol}{\mathbb{E}}

\newcommand\R{\mathbb R}
\newcommand\N{\mathbb N}
\newcommand\Z{\mathbb Z}
\newcommand\Q{\mathbb Q}

\newcommand\Oh{O}
\def\.{\mskip1mu}

\DeclareMathOperator\ld{ld}

\newcommand{\ProbSymbol}{\ensuremath{\mathbb{P}}}
\providecommand{\given}{}
\DeclarePairedDelimiterXPP\Prob[1]{\ProbSymbol}[]{}{%
	\renewcommand\given{\nonscript\:\delimsize\vert\nonscript\:\mathopen{}}%
	#1%
}
\DeclarePairedDelimiterXPP\E[1]{\ESymbol}[]{}{%
	\renewcommand\given{\nonscript\:\delimsize\vert\nonscript\:\mathopen{}}%
	#1%
}
\DeclarePairedDelimiterXPP\Eover[2]{\ESymbol_{#1}}[]{}{%
	\renewcommand\given{\nonscript\:\delimsize\vert\nonscript\:\mathopen{}}%
	#2%
}
\providecommand{\Prob}{} 
\providecommand{\E}{} 
\providecommand{\Eover}{} 

\newcommand\ui[2]{#1^{\smash{(}#2\smash{)}}}

\usepackage{fixmath}

\newcommand\eqdist{	
	\mathchoice{
		\mathrel{\overset{\raisebox{0ex}{$\scriptstyle \cal D$}}=}%
	}{
		\mathrel{\like{=}{%
			\overset{\raisebox{-1ex}{$\scriptscriptstyle \cal D$}}=%
		}}%
	}{
		\mathrel{\overset{\cal D}=}%
	}{
		\mathrel{\overset{\cal D}=}%
	}%
}

\newcommand\like[3][c]{%
	\mathchoice{
		\makeboxlike[#1]{%
			\ensuremath{\displaystyle\relax#2}%
		}{%
			\ensuremath{\displaystyle\relax#3}%
		}%
	}{
		\makeboxlike[#1]{%
			\ensuremath{\textstyle\relax#2}%
		}{%
			\ensuremath{\textstyle\relax#3}%
		}%
	}{
		\makeboxlike[#1]{%
			\ensuremath{\scriptstyle\relax#2}%
		}{%
			\ensuremath{\scriptstyle\relax#3}%
		}%
	}{
		\makeboxlike[#1]{%
			\ensuremath{\scriptscriptstyle\relax#2}%
		}{%
			\ensuremath{\scriptscriptstyle\relax#3}%
		}%
	}
}

\newcommand\uniform{\mathcal U}

\newcommand\bernoulli{\mathrm B}

\newcommand\multinomial{\mathrm{Mult}}
\newcommand\binomial{\mathrm{Bin}}

\newcommand\betadist{\mathrm{Beta}}

\newcommand\betaBinomial{\mathrm{BetaBin}}
\newcommand\dirichletMultinomial{\mathrm{DirMult}}

\DeclarePairedDelimiterXPP\distFromWeights[1]{\mathcal D}(){}{#1}
\providecommand\distFromWeights{} 

\newcommand\BetaFun{\mathrm B}

\renewcommand\given{\mathbin{\mid}}
\newcommand\harm[1]{\ensuremath{H_{#1}}}

\newcommand\ce{\colonequals}

\newcommand{\surroundedmath}[3]{
	\mathchoice{
		#1{#2{#3}#2}%
	}{
		#1{#3}%
	}{
		#1{#3}%
	}{
		#1{#3}%
	}%
}

\newcommand\rel[1]{\surroundedmath{\mathrel}{\:}{#1}}
\newcommand\wrel[1]{\surroundedmath{\mathrel}{\;}{#1}}
\newcommand\wwrel[1]{\surroundedmath{\mathrel}{\;\;}{#1}}
\newcommand\bin[1]{\surroundedmath{\mathbin}{\:}{#1}}
\newcommand\wbin[1]{\surroundedmath{\mathbin}{\;}{#1}}

\newcommand{\eqwithref}[2][c]{%
	\relwithref[#1]{#2}{=}%
}
\newcommand{\relwithref}[3][c]{%
	\mathrel{\underset{\mathclap{\makebox[\widthof{$=$}][#1]{\scriptsize\wref{#2}}}}{#3}}%
}
\newcommand{\relwithtext}[3][c]{%
	\mathrel{\underset{\mathclap{\makebox[\widthof{$=$}][#1]{\scriptsize#2}}}{#3}}%
}

\newcommand\indicator[1]{\indicatornobraces{\{#1\}}}
\newcommand\indicatornobraces[1]{\mathds 1_{#1}}
\newcommand\characteristicvector[2][]{%
	\ifthenelse{\equal{#1}{}}{%
		\mathds1_{#2}%
	}{%
		\ui{\mathds1_{#2}}{#1}%
	}
}

\newcommand\entropy[1][]{%
	\ensuremath{%
		\mathcal H_{#1}%
	}\xspace%
}

\newcommand\discreteEntropy[1][]{%
	\ensuremath{
		\entropy[] \ifthenelse{\equal{#1}{}}{ }{ (#1) }
	}\xspace%
}

\newcommand\contentropy[1][]{%
	\ensuremath{%
		\entropy[\ln]%
		\ifthenelse{\equal{#1}{}}{ }{ (#1) }%
	}\xspace%
}

\newcommand\binaryEntropy[1][]{%
	\ensuremath{
		\entropy[\ld] \ifthenelse{\equal{#1}{}}{ }{ (#1) }%
	}\xspace%
}

\newcommand\convAS{	
	\mathchoice{
		\mathrel{\overset{\raisebox{0ex}{\text{a.s.}}}\longrightarrow}%
	}{
		\mathrel{\like{\longrightarrow}{%
			\overset{\raisebox{-1ex}{\kern-1pt\scriptsize\text{a.s.}}}\longrightarrow%
		}}%
	}{
		\mathrel{\overset{\text{a.s.}}\longrightarrow}%
	}{
		\mathrel{\overset{\text{a.s.}}\longrightarrow}%
	}%
}

\newcommand\convD{	
	\mathchoice{
		\mathrel{\overset{\raisebox{0ex}{$\scriptstyle \cal D$}}\longrightarrow}%
	}{
		\mathrel{\like{\longrightarrow}{%
			\overset{\raisebox{-1ex}{$\scriptscriptstyle \cal D$}}\longrightarrow%
		}}%
	}{
		\mathrel{\overset{\cal D}\longrightarrow}%
	}{
		\mathrel{\overset{\cal D}\longrightarrow}%
	}%
}

\newcommand\ins{\mathrm{ins}}
\newcommand\del{\mathrm{del}}

\usepackage{manfnt}

\DeclarePairedDelimiterXPP\dirichletExpectation[2]{\ESymbol_{D(#2)}}[]{}{%
	\renewcommand\given{\nonscript\:\delimsize\vert\nonscript\:\mathopen{}}%
	#1%
}
\providecommand\dirichletExpectation{} 

\let\oldalign\align
\let\endoldalign\endalign

\newcommand*{\numberthis}{\stepcounter{equation}\tag{\theequation}}

\renewenvironment{align}{%
	\begingroup%
	\let\oldhalign\halign
	\def\halign{%
		\let\oldbreak\\%
		\def\nonnumberbreak{\nonumber\oldbreak*}%
		\def\\{%
			\@ifstar{\nonnumberbreak}{\oldbreak}%
		}%
		\oldhalign%
	}
	\oldalign%
}{%
	\endoldalign%
	\endgroup%
}

\newcommand{\Li}{\ensuremath{\mathcal{J}}\xspace}

\newlength\paragraphhspace
\let\oldparagraph\paragraph
\renewcommand\paragraph{%
    \@ifstar{\myparagraphStar}{\myparagraphNoStar}%
}
\newcommand\myparagraphStar[1]{%
	\oldparagraph*{#1.\hspace{\paragraphhspace}}%
}
\newcommand\myparagraphNoStar[2][]{%
	\ifthenelse{\equal{#1}{}}{%
		\oldparagraph[#2]{#2.\hspace{\paragraphhspace}}%
	}{%
		\oldparagraph[#1]{#2.\hspace{\paragraphhspace}}%
	}%
}

\ifsiam{
	\newlength\subsectionhspace
	\setlength\subsectionhspace{0pt plus 1em}
	
	\let\oldsubsection\subsection
	\renewcommand\subsection{%
	    \@ifstar{\mysubsectionStar}{\mysubsectionNoStar}%
	}
	\newcommand\mysubsectionStar[1]{%
		\oldsubsection*{#1.\hspace{\subsectionhspace}}%
	}
	\newcommand\mysubsectionNoStar[2][]{%
		\ifthenelse{\equal{#1}{}}{%
			\oldsubsection[#2]{#2.\hspace{\subsectionhspace}}%
		}{%
			\oldsubsection[#1]{#2.\hspace{\subsectionhspace}}%
		}%
	}

}{}

\colorlet{symmetriccolor}{black!10}
\colorlet{grayedout}{black!50}

\colorlet{nodecolor}{blue!70!black}
\colorlet{backbonecolor}{gray}
\colorlet{jumpcolor}{red!50!black}
\colorlet{endpointercolor}{green!50!black}

\tikzset{
	inner node/.style={
		circle,draw,inner sep=1pt,
		minimum size=1.7em,
		fill=backgroundshading!50,
	},
	leaf node/.style={
		draw,rectangle,minimum size=1.25em,inner sep=1pt,
		fill=backgroundshading!75,
	},
	inner node concise/.style={
		circle,draw,fill,minimum size=.8ex,inner sep=0pt,
	},
	leaf node concise/.style={
		draw,fill,rectangle,minimum size=.5ex,inner sep=0pt,
	},
	leaf node empty/.style={
		draw=none,fill=none,rectangle,minimum height=.5ex,minimum width=.15ex,
	},
	partitioning node/.style={
		draw,inner sep=1pt,minimum height=1.5em,minimum width=1.5em,
		rectangle,rounded corners=2pt,
		fill=backgroundshading!50,
	},
	graphs/comparison tree detailed/.style={
		extended binary tree layout,
		nodes={inner node},
		leaf/.style={leaf node},
		math nodes,
		level distance=0ex, sibling distance=0ex,
		sibling sep=.6em, level sep=.4em,
	},
	graphs/comparison tree detailed compact/.style={
		extended binary tree layout,
		nodes={inner node,minimum size=1.05em,inner sep=0pt},
		leaf/.style={leaf node,minimum size=.6em,font=\scriptsize},
		math nodes,
		level distance=0ex, sibling distance=0ex,
		sibling sep=.6em, level sep=.4em,
	},
	graphs/recursion tree detailed/.style={
		tree layout,
		significant sep=.5em,
		missing nodes get space,
		nodes={partitioning node},
		leaf/.style={partitioning node},
		math nodes,
		level distance=0ex, sibling distance=0ex,
		sibling sep=.8em, level sep=.8em,
	},
	graphs/ternary comparison tree detailed/.style={
		tree layout,
		significant sep=.5em,
		missing nodes get space,
		nodes={inner node,minimum size=1.8em,},
		leaf/.style={leaf node},
		math nodes,
		edge quotes={
			pos=.33,
			draw,circle,fill=backgroundshading,
			inner sep=0pt,
			minimum size=8pt,
			font=\tiny,
			anchor=center,
		},
		level distance=0ex, sibling distance=0ex,
		sibling sep=1.1em, level sep=1.3em,
	},
	graphs/comparison tree concise/.style={
		binary tree layout,
		significant sep=0ex,
		nodes={inner node concise},
		empty nodes,
		leaf/.style={leaf node empty},
		level distance=0ex, sibling distance=0ex,
		sibling sep=.8ex, level sep=.4ex,
	},
}

\tikzset{
	sn/.style={draw,nodecolor,thick, minimum size=5mm,inner sep=1pt},
	ino/.style={draw=none,inner sep=2pt},
	new/.style={sn,fill=red,minimum width=3mm},
	backbone/.style={thick,->,backbonecolor},
	jumppointer/.style={very thick,->,jumpcolor},
	jumppointer bend/.style={jumppointer,bend left},
	endpointer/.style={semithick,dotted,->,endpointercolor},
	minBST internal/.style={sn,circle},
	minBST min/.style={
		sn,circle,
		minimum size=3mm,inner sep=1.5pt,
		green!50!black,fill=green!15,
		font=\footnotesize,
	},
	minBST leaf/.style={
		chamfered rectangle,
		chamfered rectangle sep=.7pt,
		sn,minimum size=4mm,minimum width=5.5mm,backbonecolor!80!black,
	},
	minBST leaf empty/.style={minBST leaf},
	minBST min edge/.style={green!50!black},
	minBST left edge/.style={->,thick,black},
	minBST right edge/.style={->,thick,jumpcolor},
}

\def\mydots{\xleaders\hbox to.5em{\hfill.\hfill}\hfill}

\newlength\tmpLenNotations
\newenvironment{notations}[1][10em]{%
	\small
	\newcommand\notationentry[1]{%
		\settowidth\tmpLenNotations{##1}%
		\ifthenelse{\lengthtest{\tmpLenNotations > \labelwidth}}{%
			\parbox[b]{\labelwidth}{%
				\makebox[0pt][l]{##1}\\%
			}%
		}{%
			\mbox{##1}%
		}%
		\mydots\relax%
	}%
	\begin{list}{}{%
		\setlength\labelsep{0em}%
		\setlength\labelwidth{#1}%
		\setlength\leftmargin{\labelwidth+\labelsep+1em}%
	}
	\ifdraft{%
		\newcommand\notation[1]{%
			\item[{##1}] \marginpar{%
				\adjustbox{scale={.55}{1},outer}{%
					\color{brown!80}%
					\scriptsize$\triangleright$\;\tiny\texttt{\detokenize{##1}}%
				}%
			}%
		}%
	}{%
		\newcommand\notation[1]{\item[{##1}]}%
	}%
	\raggedright
}{%
	\end{list}
}

\let\oldthebibliography\thebibliography
\renewcommand\thebibliography[1]{%
	\oldthebibliography{#1}%
	\pdfbookmark[1]{References}{}%
	\markright{References}
	\let\path\nolinkurl%
}

\newsavebox\tabletmp%
\newlength\tmpheight%
\newlength\tmpdepth%

\let\epsilon\varepsilon
\def\DMTvarEq{=}

\ifsiam{

	\ifdraft{\pagestyle{plain}}{}
}{}

\counterwithout{equation}{section}

\ifdraft{%
	\overfullrule=10mm
	
	\usepackage[color,notref,notcite]{showkeys}
	\definecolor{refkey}{gray}{.5}
	\colorlet{labelkey}{green!50!black!50}
	
	\usepackage[inline,nolabel]{showlabels}

	\showlabels{cite}
	\showlabels{citealt}
	\showlabels{citealp}
	\showlabels{citet}
	\showlabels{Citet}
	\showlabels{citep}
	\showlabels{citeauthor}
	\showlabels{Citeauthor}
	\showlabels{citefullauthor}
	\showlabels{citeyear}
	\showlabels{citeyearpar}
	\showlabels{wref}
	\showlabels{wpref}
	\showlabels{wtpref}
	\showlabels{wildref}
	\showlabels{wildpageref}
	\showlabels{wildtpageref}
}{}

\newsavebox\citehrefbox
\newsavebox\tmpbox

\newcommand\extendedversion{extended online version\xspace}

%% file: jumplists-appendix.tex
%
%
%

\clearpage
\appendix

\onecolumn
\normalsize

\ifsubmission{
	\KOMAoptions{headings=big}
	\addtokomafont{paragraph}{\color{black!65}}
	\RedeclareSectionCommand[
		beforeskip=-1.5\bigskipamount,
		afterskip=1.\bigskipamount,
	]{section}
	\RedeclareSectionCommand[
		beforeskip=-1\bigskipamount,
		afterskip=1\medskipamount,
	]{subsection}
	\typearea{10}
	\raggedbottom
}{}

\numberwithin{algorithm}{section}
\numberwithin{table}{section}
\numberwithin{figure}{section}
\numberwithin{equation}{section}

\part*{Appendix}
\pdfbookmark[0]{Appendix}{}
\ifsiam{}{%
	\manualmark%
	\markleft{\mytitle}%
	\markright{Appendix}
}

\input{jumplists-notation}

\section{Comparison of Jumplist Definitions}

\label{app:differences-definitions}

Our definition of jumplists differs in some details from the original version.
We list the differences here, and discuss why we think that our modifications are appropriate.

\paragraph{Symmetry}

In the original version of the jumplist, the jump pointer is allowed to target any node 
from the sublist, except the header itself. 
Thus there are $m-1$ possible choices. 
In this setting, the size of the next-sublist can attain any value between $0$ and $m-2$, 
whereas the size of the jump-sublist is between $1$ and $m-1$.

We disallow the direct successor of the head as possible target. 
This modification restores symmetry between next- and jump-sublist: both must be non-empty and 
contain at most $m-2$ nodes and their sizes have the same distribution.
Moreover, forbidding the direct successor as jump target is also a natural
requirement since such a degenerate ``shortcut'' is useless in searches. 

\paragraph{Small Sublists}

The original jumplists only have one type of nodes which corresponds to our jump node. 
In the case $m=1$, Brönnimann, Cazals, and Durand resort to assigning an 
``exceptional pointer'' to the direct successor;
note that this node actually lies outside (one behind) of the current sublist. 
These pointers are of no use, as they are never followed during (jump-and-walk) search. 

In implementations with heap-allocated memory for each node,
it is often not a problem to have different node types (and sizes),
and it potentially allows to save memory.
We thus introduced the plain node without jump pointer, 
used whenever the sublist has at most $w$ nodes.
$w\ge 2$ is required if we want to avoid useless jump pointers that point to the direct successor.

This also allows us enforce that every node has at most one incoming jump pointer;
this is another natural requirement from the perspective of a search
starting at the header:
shortcuts with the same target are redundant.
The parameter $w$ allows us to trade space for time.

\paragraph{Sentinel vs.\ Circularly closed}

The original jumplist implementation has a circularly closed backbone, \ie, 
the next pointer of the last node in the list points to the overall header again,
avoiding special treatment for an empty list.
Since the backbone is sorted, we can instead 
add a \emph{sentinel} node with key $+\infty$ at the end of the list, 
so we can omit any explicit boundary checks during searches.

\input{jumplists-algorithms}

\section{Omitted proofs}
\label{app:expectations}

\begin{proof}[\wref{lem:binomial-negative-factorial-moments}]
For $m\in\{1,2\}$, we compute
\begin{align*}
		\E{ X^{\underline {-m}} }
	&\wwrel=
		\sum_{x=0}^n 
			x^{\underline{-m}}\cdot
			\binom{n}{x} \;
			p^{x} q^{n-x}
	\\ &\wwrel=
		n^{\underline {-m}} \, p^{-m} 
		\sum_{x=0}^{n}
			\binom{n+m}{x+m} \;
			p^{x+m} q^{(n+m)-(x+m)}
	\\ &\wwrel=
		n^{\underline {-m}} \, p^{-m} 
			\sum_{x=m}^{n+m} \binom{n+m}{x} p^{x} q^{(n+m)-x}
	\\ &\wwrel{\relwithtext[r]{[binom.\,thm.]}=}
		n^{\underline {-m}} \, p^{-m} 
		\;  
		\biggl((\underbrace{p+q} _ {=1} )^{n+m}  - 
			\sum_{x=0}^{m-1} \binom{n+m}{x}  p^{x} q^{(n+m)-x}
		 \biggr)
		.
\end{align*}
For the first part of the claim, we set $m=1$ and find that the sum reduces
to $q^{n+1}$;
for the second part of the claim, we use $m=2$ and note that the expression
in the outer parentheses is at most $1$.
\end{proof}

\begin{proof}[\wref{lem:E-ln-D}]
We use the following known integral;
see \cite[\href{https://www.wild-inter.net/publications/html/wild-2016.pdf.html\#pf46}{Eq.\,(2.30)}]{Wild2016}:
\begin{align*}
\numberthis\label{eq:logarithmic-beta-integral}
		\int_0^1 z^{a-1} (1-z)^{b-1} \ln(z) \, dz
	\wwrel=
		\BetaFun(a,b) \bigl( \psi(a) - \psi(a+b) \bigr)
		,\qquad (a,b > 0) .
\end{align*}
Here $\psi(z) = \frac d{dz} \ln(\Gamma(z))$ is the digamma function. 
Then we find
\begin{align*}
		\E{\ln(D)}
	&\wwrel=
		\int_0^1 \ln(x)  \frac{x^{t}(1-x)^{t}}
				{\BetaFun(t+1,t+1)} \, dx
\\	&\wwrel{\eqwithref{eq:logarithmic-beta-integral}}
			\psi(t+1)-\psi(k+1)
\\	&\wwrel=
		\harm{t} - \harm{k} 
		,
\end{align*}
and
\begin{align*}
		\E{D \ln(D)}
	&\wwrel=
		\int_0^1 x \ln(x)  \frac{x^{t}(1-x)^{t}}
				{\BetaFun(t+1,t+1)} \, dx
\\	&\wwrel=
		\frac{\BetaFun(t+2,t+1)}{\BetaFun(t+1,t+1)} 
			\int_0^1 \ln(x)  \frac{x^{t+1}(1-x)^{t}}
				{\BetaFun(t+2,t+1)} \, dx
\\	&\wwrel{\eqwithref{eq:logarithmic-beta-integral}}
		\frac{t+1}{k+1} 
			\bigl(\psi(t+2)-\psi(k+2)\bigr)
\\	&\wwrel=
		\frac{t+1}{2t+2} 
		\bigl(\harm{t+1} - \harm{k+1}\bigr)
\\	&\wwrel=
		\frac12\bigl(\harm{t+1} - \harm{k+1}\bigr) \;.
\end{align*}
\end{proof}

%% file: jumplists-notation.tex
\section{Index of Notation}
\label{app:notations}

In this appendix, we collect the notations used in this work.

\newcommand\ilmenau[2][]{%
}

\newlength\notationwidth
\setlength\notationwidth{9em}

\subsection{Generic Mathematical Notation}
\begin{notations}[\notationwidth]
\notation{$\N$, $\N_0$, $\Z$, $\Q$, $\R$}
	natural numbers $\N = \{1,2,3,\ldots\}$, 
	$\N_0 = \N \cup \{0\}$,
	integers $\Z = \{\ldots,-2,-1,0,1,2,\ldots\}$,
	rational numbers $\Q$,
	real numbers $\R$.
\notation{$\R_{>1}$, $\N_{\ge3}$ etc.}
	restricted sets $X_\mathrm{pred} = \{x\in X : x \text{ fulfills } \mathrm{pred} \}$.
\notation{$0.\overline 3$}
	repeating decimal; $0.\overline3 = 0.333\ldots = \frac13$; \\
	numerals under the line form the repeated part of 
	the decimal number.
\notation{$\ln(n)$, $\ld(n)$}
	natural and binary logarithm; $\ln(n) = \log_e(n)$, $\ld(n) = \log_2(n)$.
\notation{$X$}
	to emphasize that $X$ is a random variable it is Capitalized.
\notation{$[a,b)$}
	real intervals, the end points with round parentheses are excluded, 
	those with square brackets are included.
\notation{$[m..n]$, $[n]$}
	integer intervals, $[m..n] = \{m,m+1,\ldots,n\}$;
	$[n] = [1..n]$.
\notation{$[\text{stmt}]$, $[x=y]$}
	Iverson bracket, $[\text{stmt}] = 1$ if stmt is true, $[\text{stmt}] = 0$ otherwise.
\notation{$\harm n$}
	$n$th harmonic number; $\harm n = \sum_{i=1}^n 1/i$.
\notation{$\Oh(f(n))$, $\pm\Oh(f(n))$, $\Omega$, $\Theta$, $\sim$}
	asymptotic notation as defined, \eg, by \cite[Section A.2]{Flajolet2009};
	$f=g\pm\Oh(h)$ is equivalent to $|f-g| \in \Oh(|h|)$.
\notation{$x \pm y$}
	$x$ with absolute error $|y|$; formally the interval $x \pm y = [x-|y|,x+|y|]$;
	as with $\Oh$-terms, we use ``one-way equalities'': $z=x\pm y$ instead of $z \in x \pm y$.
\notation{$\Gamma(z)$}
	the gamma function, $\Gamma(z) = \int_0^\infty t^{z-1}e^{-t} \, dt$.
\notation{$\psi(z)$}
	the digamma function, $\psi(z) = \frac d{dz} \ln(\Gamma(z))$.
\notation{$\BetaFun(\alpha,\beta)$}
	the beta function, $\BetaFun(\alpha,\beta) = \int_0^1 z^{\alpha-1}(1-z)^{\beta-1}\, dz$
\notation{$a^{\underline b}$, $a^{\overline b}$}
	factorial powers notation of Graham et al.~\cite{ConcreteMathematics}; 
	``$a$ to the $b$ falling resp.\ rising.''
\notation{$h(x)$}
	the binary base-$e$ entropy function
	$h(x) = -x\ln(x) - (1-x) \ln(1-x)$.
\end{notations}

\subsection{Stochastics-related Notation}
\begin{notations}[\notationwidth]
\notation{$\Prob{E}$, $\Prob{X=x}$}
	probability of an event $E$ resp.\ probability for random variable $X$ to
	attain value $x$.
\notation{$\E{X}$}
	expected value of $X$.
\notation{$X\eqdist Y$}
	equality in distribution; $X$ and $Y$ have the same distribution.
\notation{$\indicatornobraces{E}$, $\indicator{X\le 5}$}
	indicator variable for event $E$, \ie, $\indicatornobraces{E}$ is $1$ if $E$
	occurs and $0$ otherwise;
	$\{X\le 5\}$ denotes the event induced by the expression $X \le 5$.
\notation{$\bernoulli(p)$}
	Bernoulli distributed random variable;
	$p\in[0,1]$.
\notation{$\uniform(a,b)$}
	uniformly in $(a,b)\subset\R$ distributed random variable. 
\notation{$\uniform[a..b]$}
	discrete uniformly in $[a..b]\subset\Z$ distributed random variable. 
\notation{$\betadist(\alpha,\beta)$}
	beta distributed random variable with shape parameters $\alpha\in\R_{>0}$ and $\beta\in\R_{>0}$.
\notation{$\binomial(n,p)$}
	binomial distributed random variable with $n\in\N_0$ trials and success probability $p\in[0,1]$;
	$\binomial(1,p)\eqdist\bernoulli(p)$. 
	$X\eqdist \binomial(n,p)$ is equivalent to $(X,n-X)\eqdist\multinomial(n;p,1-p)$.
\notation{$\betaBinomial(n,\alpha,\beta)$}
	beta-binomial distributed random variable;
	$n\in\N_0$, $\alpha,\beta\in\R_{>0}$;
	$X\eqdist \betaBinomial(n,\alpha,\beta)$ is equivalent to
	$(X,n-X) \eqdist \dirichletMultinomial(n;\alpha,\beta)$.
\end{notations}

\subsection{Notation for Jumplists and Analysis}

\begin{notations}[\notationwidth]
\notation{$k$, $t$}
	sample size $k = 2t+1$, $t\in\N_{\ge0}$; 
	jump pointers are chosen as median of $k$ elements.
\notation{$w$}
	leaf size; (sub)lists with $n<w$ (equivalently: $m\le w$)
	do not use jump pointers.
\notation{$n$}
	number of keys stored; the input size.
\notation{$m$, $m(v)$}
	the number of nodes; $m=n+1$ (the header does not store a key).
\notation{$x_1,\ldots,x_n$}
	the stored keys; $x_1<\cdots<x_n$.
\notation{$v_0,v_1,\ldots,v_n$}
	the $m=n+1$ nodes of a jumplist on $n$ keys, in the order of the backbone,
	\ie, $v_i.\id{key} = x_i$ and $v_{i-1}.\id{next} = v_i$, $i=1,\ldots,n$.
\notation{$\mathcal J_n$}
	random jumplist on $n$ keys $\{1,\ldots,n\}$; obtained from \proc{Rebalance}
	on a list with keys $\{1,\ldots,n\}$.
\notation{sublist of node $v_i$}
	the sublist that starts at $v_i$ (inclusive) and extends up to (excluding) the
	first node targeted by a jump pointer of a node $v_j$ with $j<i$ or up to (including) the
	end of the whole list if not such pointer exists.
\notation{$\mathcal J_1$, $\mathcal J_1(v)$}
	the next-sublist (of a given node $v$); the sublist of $v.\mathit{next}$;
	(only defined for jump nodes).
\notation{$\mathcal J_2$, $\mathcal J_2(v)$}
	the jump-sublist (of a given node $v$); the sublist of $v.\mathit{jump}$.	
	(only defined for jump nodes).
\notation{$J_1$, $J_2$}
	(random) sublist sizes;
	$J_r = I_r + t+1 \in [t+1..m-t-2]$ is the number of nodes in $\mathcal J_r$, $r\in\{1,2\}$;
	$J_1+J_2 = m-1$;
\notation{$I_1$, $I_2$}
	$I_1\eqdist I_2 \eqdist \betaBinomial(m-k-2,t+1,t+1)$.
\end{notations}

%% file: jumplists-algorithms.tex
\section{Algorithms}
\label{app:algorithms}

In this appendix, we give the more details the insertion and deletion algorithms 
for randomized median-of-$k$ jumplists. 

We describe the procedures in prose and an intuitive graphic syntax,
as well as in detailed pseudocode; see \wref{sec:pseudocode} for the latter.
We also point out that our proof-of-concept implementation in Java
is available online for interested readers~\cite{codeWild2016}.

As a simple example to introduce the graphical syntax, 
here is the transformation from 
jumplists to dangling-min BSTs pictorially:
\medskip

\noindent\plaincenter{%
\begin{tikzpicture}[
		node distance=1.5cm,
		every node/.style={inner sep=1.5pt},
	]
	\node[anchor=east] (6)  at (-.35,0){$\proc{minBST}\biggl($};
	\node[sn] (0) { };
	\node[sn] (1)  at (1,0){$x_1$};
	\node[sn,anchor=west,xshift=-.8pt] (J1) at (1.east) {\phantom{xxxxxxxxx}}; 
	\node (3)  at (2.7,0){ };
	\node[sn] (2)  at (3.5,0){$x_j$};
	\node[sn,anchor=west,xshift=-.8pt] (J2) at (2.east) {\phantom{xxxxxxxxxxxx}}; 
	\node[anchor=west] at (5.75,0) {$\biggr)$} ;
	\node at (6.55,0) {$=$} ;
	\path[backbone]
		(0) edge (1)
		(J1) edge (2);
	\draw[jumppointer]
		(0) to[out=45,looseness=.5,in=120] (2);
	
	\draw[thick, decorate, decoration={ brace, mirror, raise=0.4cm}] (1.180) -- 
		node[below=.55cm] {$\mathcal{J}_1$} (J1.360);  
	\draw[thick, decorate, decoration={ brace, mirror, raise=0.4cm}] (2.180) -- 
		node[below=.55cm] {$\mathcal{J}_2$} (J2.360);
	
	\begin{scope}[shift={(8.5,.5)}]
		\node[minBST internal] (0) {$\vphantom| x_j$};
		\node[minBST min] (m0) at ($(0)+(195:1.75em)$) {$x_1$};
		\draw[minBST min edge] (0) -- (m0) ;
		\node (2)  [ below left of=0, yshift = -0.2cm] {$\proc{minBST}(\mathcal{J}_1)$};
		\node (3) [below right  of=0, yshift = -0.2cm]{$\proc{minBST}(\mathcal{J}_2)$};
		\draw[minBST left edge] (0) -- (2);
		\draw[minBST right edge] (0) -- (3);
	\end{scope}
	
		\begin{scope}[shift={(0,1.4)}]
			\node[sn] (h1) at (3.4,0){ };
			\node[sn,anchor=west,xshift=10pt] (L) at (h1.east) {$\;x_1 \;\ldots\;x_n\;$};
			\node[anchor=east]  at (3.1,0){$\proc{minBST}\biggl($};
			\node[anchor=west] at (5.75,0) {$\biggr)$} ;
			\draw[backbone] (h1) -- (L) ;
			\node at (6.55,0) {$=$} ;
			\node[minBST leaf] at (8,0) {$\;x_1,\ldots,x_n\;$} ;
		\end{scope}
	
	\node[anchor=west] at (10.5,0) {$(m>w)$} ;
	\node[anchor=west] at (10.5,1.4) {$(m\le w)$} ;
\end{tikzpicture}%
}
\smallskip

\noindent
The first equation defines \proc{minBST} on small jumplists ($m\le w$);
it shows a header without jump pointer, \ie, a plain node.
The second equation defines \proc{minBST} on larger jumplists.
Whenever variables appear on the left side, 
they are understood as formal placeholders of a pattern to be matched against
the actual input. 
This mimics the corresponding feature of many functional programming languages
that allows to define a function case by case in this syntax.
The parts that match the variables are then used on the right-hand side.

\paragraph{Graphical syntax conventions}
We now proceed to the description of the insertion and deletion procedures.
We use the following conventions:
	The input of the algorithms, the ``old'' jumplist, 
	is drawn as rectangle (or abbreviated by $\mathcal{J}$).
	A sequence of (output) leaf nodes is depicted by a rectangle with rounded corners.
	The position of insertion resp.\ deletion is marked in red.
	If the algorithm makes a random choice,
	each outcome is multiplied with its probability, and all outcomes are added up.

\subsection{Rebalance}

Algorithm \proc{Rebalance} is used if a jumplist needs to be (re)built from scratch. 
It only uses the backbone of the argument, any existing jump pointers are ignored.
In the base case, \ie, if the argument contains $m\leq w$ nodes, 
a linked list of plain nodes with the same keys is returned.

\smallskip

\plaincenter{%
\begin{tikzpicture}[
		every node/.style={inner sep=1.5pt,anchor=west},
	]
	\node (head) {$\proc{Reb}\biggl($};
	\node[sn] (J) at (head.east) { $\;v_0 \;\ldots\; v_n\;$};
	\node (headend) at (J.east) {$\biggr)$} ;
	\path (headend.east) --node[xshift=-.2em,anchor=center] {$=$} ++(3em,0) 
		node[sn,rounded corners] (leaf) {$\;v_0 \;\ldots\; v_n\;$ } ;
	\node[xshift=2em] at (leaf.east) {$(m\le w)$} ;
\end{tikzpicture}%
}

\noindent
If \Li contains $m>w$ nodes, $v_0$ must become a jump node and we have to draw 
a jump target from the sample range. 
Conceptually, a sample of $k$ nodes is drawn and the median \wrt the keys is chosen.
The same distribution can actually be achieved without explicitly drawing samples 
using a random variable $J \eqdist \betaBinomial(m-2-k,t+1,t+1)+t+2$ 
(see \wref{sec:median-of-k-jumplists-def}).
Then the node $v_J$ is the jump target.
After the jump pointer of $v_0$ has been initialized, 
the resulting next- and jump-sublist are rebalanced recursively.

\plaincenter{%
\begin{tikzpicture}[
		every node/.style={inner sep=1.5pt,anchor=west},
	]
	\node (head) {$\proc{Reb}\biggl($};
	\node[sn] (J) at (head.east) { $\;v_0 \;\ldots\; v_n\;$};
	\node (headend) at (J.east) {$\biggr)$} ;

	\path (headend.east) --node[xshift=-.2em,anchor=center] {$=$} ++(3em,0) 
			node[anchor=west,sn] (v0) {$v_0$};
	\node (3)  [right of = v0] {$\proc{Reb} \biggl($};
	\node[sn] (J1) at (3.east) {\phantom{xxxxxxx}}; 
	\node (4)  at (J1.east){$\biggl)$ };
	
	\node (5) [xshift=1em] at (4.east) {$\proc{Reb} \biggl($};
	\node[sn] (6) at (5.east){$v_\mathcal{J}$};
	\node[sn,xshift=-.8pt] (J2) at (6.east) {\phantom{xxxxxxxxx}}; 
	\node (7) [color=nodecolor,anchor=east]  at (J2.east) {$v_n\,$};
	\node (8) at (J2.east) {$\biggr)$} ;
	
	\path[backbone,on layer=background]
		(v0) edge[shorten >= -2pt] (3)
		(4) edge[shorten >=-2pt,shorten <=-4pt] (5);
	\draw[jumppointer]
		(v0) to[out=45,looseness=.4,in=90] (6);
	\node[xshift=2em] at (8.east) {$(m> w)$} ;
\end{tikzpicture}%
}

\subsection{Insert}

\proc{Insert} in jumplists consists of three phases found in many tree-based
dictionaries: 
(unsuccessful) search, insertion, and cleanup.
Unless $x$ is already present,
the search ends at the node with the largest key (strictly) smaller than $x$.
There we insert a new node with key $x$ into the backbone.

The new node however does not have a jump pointer yet. 
Furthermore, the new node might need to be considered as potential jump target
of its predecessors in the backbone.
Thus, for all the nodes that have the new node in their sublist, 
we need to restore the pointer distribution. 
This is carried out by \proc{RestoreAfterInsert}.

Let $m$ be the number of nodes after the insertion, \ie, including the new node.
If $m \le w$, the new node remains a plain node within a list of plain nodes,
and no cleanup is necessary.
If $m=w+1$ due to the insertion, $v_0$, which was a plain node before, 
now has to become a jump node. 
In this case, \proc{Rebalance} is called on \Li and the insertion terminates.
\smallskip

	\plaincenter{%
	\begin{tikzpicture}[
			every node/.style={inner sep=1.5pt,anchor=west},
		]
		\node (head) {$\proc{RestIns}\biggl($};
		\node[new,xshift=1em] (new) at (head.east) {};
		\node[sn] (J) at (head.east) { \phantom{xxxxxxxxxx}};
		\node (headend) at (J.east) {$\biggr)$} ;
		
		\path (headend.east) --node[anchor=center] {$=$} ++(3em,0) 
			node (2) {$\proc{Reb}\biggr($} ;
		\node[new,xshift=1em] (new2) at (2.east) {};
		\node[sn] (Jp) at (2.east) {\phantom{xxxxxxxxxx} } ;
		\node[anchor=west, ] (2) at (Jp.east) {$\biggr)$} ;
		\node[xshift=2em] at (2.east) {$(m\le w)$} ;
	\end{tikzpicture}%
	}

\noindent
If $m > w+1$, we first restore the pointer distribution of $v_0$. 
Due to the insertion of a new node, the sample range now contains an additional node $u$.
Note that $u$ is not necessarily the newly inserted node; 
if the new key is the first or second smallest in \Li, 
$u$ is the former second node of \Li.

If we, conceptually, wanted to draw pointers for \Li anew,
there are two possibilities: 
either $u$ is part of the sample, or $u$ is not part of the sample.
The probability for the latter case is
\begin{align}
		\binom{n-1}k \bigg/ \binom nk
	\wwrel=
		\frac{(n-1)! k! (n-k)!}{k!(n-k-1)! n!}
	\wwrel=
		\frac{n-k}{n}
	\wwrel=
		1 - \frac kn,
\end{align}
since the overall number of $k$-samples from $n$ items is $\binom nk$ and
if we forbid, say, item $n$, we have $\binom {n-1}k$ choices left.

Let us denote by $p = k/n$ the counter probability, \ie, the probability
that $u$ is part of the sample.
In that case, we have to rebalance all of \Li.
Conditional on the event that $u$ is not in the sample,
the existing jump pointer of $v_0$ has the correct distribution:
it has been chosen as the median of a random sample not containing $u$.

In the algorithm, we thus rebalance \Li with probability $p$, 
where we draw the jump pointer of $v_0$ conditional on $u$ being part of the sample.
Otherwise, $v_0$'s jump pointer can be kept, and we continue recursively
in the uniquely determined sublist that contains the inserted node
\parenthesissign that is, unless $v_0$ does not have a jump pointer yet, 
since $v_0$ is the newly inserted node.
In that case, we simply \emph{steal} the jump pointer of its direct successor, $v_1$,
which has the correct conditional distribution.
Now $v_1$ does not have a jump pointer, and we treat this case recursively, as if $v_1$
was the newly inserted node.

\noindent
\plaincenter{%
\adjustbox{max width=\linewidth}{\begin{tikzpicture}[
		node distance=1cm,
		every node/.style={inner sep=1.5pt,anchor=west},
	]
	\node (head) at (-1.5,0) {$\proc{RestIns}(\Li)$};
	\path (head.east) --node[anchor=center] {$=$} ++(3em,0) 
			node (p) {$p \cdot {}$} ;
			
	\node[sn] (v0) at (p.east) {$v_0$};
	\node (3)  [right of = v0] {$\proc{Reb} \biggl($};
	\node[sn] (J1) at (3.east) {\phantom{xxxxxxx}}; 
	\node (4)  at (J1.east){$\biggl)$ };
	
	\node (5) [xshift=1em] at (4.east) {$\proc{Reb} \biggl($};
	\node[sn] (6) at (5.east){$v_{J}$};
	\node[sn,xshift=-.8pt] (J2) at (6.east) {\phantom{xxxxxxxxx}}; 
	\node (7) [color=nodecolor,anchor=east]  at (J2.east) {$v_n\,$};
	\node at (J2.east) {$\biggr)$} ;
	
	\path[backbone,on layer=background]
		(v0) edge[shorten >=-2pt] (3)
		(4) edge[shorten >=-2pt,shorten <=-4pt] (5);
	\draw[jumppointer]
		(v0) to[out=45,looseness=.4,in=90] (6);

	\begin{scope}[shift={(-.3,-2.75)}]
		\node (p1)  {${}+(1-p)\cdot {}$};
		
		\node[sn,xshift=.8em] (n2) at (p1.east) {$v_0$};
		\node (n3)  [right of = n2, xshift=.7em] {$\proc{RestIns} \biggl($};
		\node (n4) [new,xshift=1.2em] at (n3.east) {};
		\node[sn] (nJ1) at (n3.east) {\phantom{xxxxxxx}}; 
		\node (n5) at (nJ1.east){$\biggl)$ };
		\node[sn] (n6) [right of = n5, xshift=-1em]  {$v_j$};
		\node[sn,anchor=west,xshift=-.6pt] (nJ2) at (n6.east) {\phantom{xxxxxx}};

		\node[sn,yshift=8.5ex] (a2) at (n2.west) {$x$};
		\node (a3)  [right of = a2, xshift=.7em] {$\proc{RestIns} \biggl($};
		\node  (a4) [new,xshift=0em] at (a3.east) {};
		\node[sn] (aJ1) at (a3.east) {\phantom{xxxxxxx}}; 
		\node (a5) at (aJ1.east){$\biggl)$};
		\node[sn] (a6) [right of = a5, xshift=-1em]  {$v_j$};
		\node[sn,xshift=-.8pt] (aJ2) at (a6.east) {\phantom{xxxxxx}};

		\node[sn,yshift=-8.5ex] (b2) at (n2.west) {$v_0$};
		\node[sn,xshift=.7em] (bJ1) at (b2.east) {\phantom{xxxxxxx}}; 
		\node[xshift=.7em] (b3) at (bJ1.east) {$\proc{restIns} \biggl($};
		\node[sn] (b4) at (b3.east) {$v_j$};	
		\node[new,xshift=1em] (b5)  at (b4.east) {};
		\node[sn,anchor=west,xshift=-.8pt] (bJ2) at (b4.east) {\phantom{xxxxxx}}; 
		\node (b6) at (bJ2.east){$\biggl)$};
		
		\path[backbone,on layer=background]
			(n2) edge[shorten >=-2pt] (n3)
			(n5) edge[shorten <=-4pt] (n6)
			(a2) edge[shorten >=-2pt] (a3)
			(a5) edge[shorten <=-4pt] (a6)
			(b2) edge (bJ1)
			(bJ1) edge[shorten >=-2pt] (b3);
		\draw[jumppointer] (n2) to[out=45,looseness=.4,in=120] (n6);
		\draw[jumppointer] (a2) to[out=45,looseness=.4,in=120] (a6);
		\draw[jumppointer] (b2) to[out=45,looseness=.45,in=90] (b4);
	
		\draw[semithick,decorate, decoration={ brace}] 
			($(b2.south west)+(-.5em,-1ex)$) -- ($(a2.north west)+(-.5em,1ex)$);	
	\end{scope}
		
	\begin{scope}[shift={(8,-2.75)}, node distance = 0.8cm]
		\node (cm1) {$\Li={}$};
		\node (cm2) [sn] at (cm1.east) {$v_0$};
		\node (cm3) [new,xshift=2.2em] at (cm2.east) {};
		\node[sn,xshift=1em] (cmJ1) at (cm2.east) {\phantom{xxxxxxx}};
		\node[sn,xshift=1em] (cm4) at (cmJ1.east) {$v_j$};
		\node[sn,anchor=west,xshift=-0.8pt] (cmJ2) at (cm4.east) {\phantom{xxxxxx}};
		
		\node[yshift=8.5ex] (ca1) at (cm1.west) {$\Li={}$};
		\node[sn,fill=red] (ca2)  at (ca1.east) {$x$};
		\node[sn,xshift=.7em,minimum width=4mm] (ca3) at (ca2.east) {$v_0$};
		\node[sn,xshift=.7em] (caJ1) at (ca3.east) {\phantom{xxxx}};
		\node[sn,xshift=.7em] (ca4) at (caJ1.east) {$v_j$};
		\node[sn,anchor=west,xshift=-0.8pt] (caJ2) at (ca4.east) {\phantom{xxxxxx}};
		
		\node[yshift=-8.5ex] (cb1) at (cm1.west) {$\Li={}$};
		\node[sn] (cb2) at (cb1.east) {$v_0$};
		\node[sn,xshift=1em] (cbJ1) at (cb2.east) {\phantom{xxxxxxx}};
		\node[sn,xshift=1em] (cb3) at (cbJ1.east) {$v_j$};
		\node[new,xshift=1em] (cb4) at (cb3.east) {};
		\node[sn,anchor=west,xshift=-0.8pt] (cbJ2) at (cb3.east) {\phantom{xxxxxx}};

		\path[backbone,on layer=background]
			(cm2) edge (cmJ1.west)
			(cmJ1.east) edge (cm4)
			(ca2) edge (ca3)
			(ca3) edge (caJ1)
			(caJ1) edge (ca4)
			(cb2) edge (cbJ1.west)
			(cbJ1.east) edge (cb3);
		\draw[jumppointer] (cm2) to[out=45,looseness=.55,in=120] (cm4);
		\draw[jumppointer] (ca3) to[out=50,looseness=.7,in=120] (ca4);	
		\draw[jumppointer] (cb2) to[out=45,looseness=.5,in=120] (cb3);	
	\end{scope}
\end{tikzpicture}%
}}

\vspace{-2ex}

\subsection{Delete}

\proc{Delete} has the same three phases as \proc{Insert}:
first a (successful) search finds the node to be deleted,
then we actually remove it from the backbone.
Finally, \proc{RestoreAfterDeletion} performs the cleanup: 
the pointer distribution for those nodes whose sublists contained the deleted node has to
be restored since their sample range has shrunk.

Let $m$ be the number of nodes after deletion, and let $u$ be the deleted node.
We first assume that $u\ne v_0$; the case of deleting $v_0$ will be addressed later.
If $m \le w-1$, \Li is a list of plain nodes and can remain unaltered.
If $m=w$, the size dropped from $w+1$ to $w$ due to the deletion, so $v_0$ has to be
made a plain node.
\smallskip

\plaincenter{%
\begin{tikzpicture}[
		every node/.style={inner sep=1.5pt,anchor=west},
	]
	\node (head) {$\proc{RestDel}\biggl($};
	\node[new,xshift=1em] (new) at (head.east) {};
	\node[sn] (J) at (head.east) { \phantom{xxxxxxxxxx}};
	\node (headend) at (J.east) {$\biggr)$} ;
	
	\path (headend.east) --node[anchor=center] {$=$} ++(3em,0) 
		node[sn,rounded corners] (leaf) {$\;v_0 \;\ldots\; v_n\;$ } ;
	\node[xshift=2em] at (leaf.east) {$(m\le w)$} ;
\end{tikzpicture}%
}

\noindent
Otherwise ($m>w$), $v_0$ is a jump node whose sublist contained $u$.
There are two possible cases:
either the sample drawn to choose $v_0.\id{jump}$ contained $u$, or not.
In the latter case, the deletion of $u$ does not affect the choice for 
$v_0.\id{jump}$ at all,
and we recursively cleanup the uniquely determined sublist that formerly contained $u$.
If $u$ was indeed part of the sample, we have to rebalance~\Li.

It remains to determine the probability $p$ that $u$ was in the sample that led to the
choice of $v_0.\id{jump}$.
Unlike for insertion, $p$ now depends on these two nodes.
Let $J_1$ resp.\ $J_2$ be the sizes of the next- resp.\ jump-sublist 
\emph{before} deletion;
recall that we store $J_1$ in $v_0.\id{nsize}$.
Then $p$ is given by the following expression:
\begin{align}
	p \wwrel= 
	\begin{dcases*}
		1, 		& if $u=v_0.\id{jump}$;\\
		\frac{t}{J_1-1},	& if $u$ was in next-sublist (where $\frac00\ce1$ in case $t=J_1-1=0$); \\
		\frac{t}{J_2-1},	& if $u$ was in jump-sublist.
	\end{dcases*}
\end{align}
The correctness is best seen in a case-by-case argument, which we give below.
But before we do that, we have to consider the case that the deleted node is $u=v_0$.
Then $v_1$ has become the new header,
but its jump pointer now has the wrong distribution since $v_0$'s jump pointer
no longer delimits its sample range.
But observe that $v_0$'s (old) sample range was exactly 
$v_1$'s new sample range plus $v_2$.
Accordingly we only have to rebalance in case $v_2$ was part of the
sample to select $v_0.\id{jump}$, which happens with probability $p=\frac{t}{J_1-1}$.
Otherwise, we can conceptually impose $v_0$'s jump pointer on $v_1$, 
which is easily implemented by swapping their keys,
and continue the cleanup recursively in the next-sublist, as if $v_1$ had been deleted.

Overall, the following situations can occur upon deletion:
\begin{enumerate}
\item If the jump pointer of $v_0$ targeted the deleted node, 
	the whole list is re-built with probability $1$.
\item If no sampling is used, \ie, $k=1$, 
	the list only needs to be reconstructed in the following two cases:
	\begin{enumerate}
	\item If the deleted node had rank $0$ and next-size $1$, 
		we cannot impose the jump pointer of the deleted node onto $v_0$ 
		as the target is not valid. 
		Thus the list is reconstructed.
	\item If the deleted node had rank $1$ and $v_0$ had next-size $1$, 
		the only node in the next-sublist has been deleted. 
		This results in an invalid pointer configuration, 
		therefore the list is reconstructed.
	\end{enumerate}
\item If the deleted node was contained in the next-sublist of $v_0$, 
	\ie, $r < J_1+1$, it was part of the sample with probability $\frac{t}{J_1-1}$
\item If the deleted node was contained in the jump-sublist of $v_0$,
	it was part of the sample with probability $\frac{t}{m-1-J_1}$.
\end{enumerate}

To conclude, depending on the outcome of the coin flip,
the algorithm either rebalances the current sublist (with probability $p$)
(as given above) and terminates, 
or it reuses the topmost old jump pointer and continues recursively.

\noindent\plaincenter{%
\adjustbox{max width=\linewidth}{%
\begin{tikzpicture}[
		node distance=1cm,
		every node/.style={inner sep=1.5pt,anchor=west},
	]
	\node (head) at (-1.5,0) {$\proc{RestDel}(\Li)$};
	\path (head.east) --node[anchor=center] {$=$} ++(3em,0) 
			node (p) {$p \cdot {}$} ;
			
	\node[sn] (v0) at (p.east) {$v_0$};
	\node (3)  [right of = v0] {$\proc{Reb} \biggl($};
	\node[sn] (J1) at (3.east) {\phantom{xxxxxxx}}; 
	\node (4)  at (J1.east){$\biggl)$ };
	
	\node (5) [xshift=1em] at (4.east) {$\proc{Reb} \biggl($};
	\node[sn] (6) at (5.east){$v_{J}$};
	\node[sn,xshift=-.8pt] (J2) at (6.east) {\phantom{xxxxxxxxx}}; 
	\node (7) [color=nodecolor,anchor=east]  at (J2.east) {$v_n\,$};
	\node at (J2.east) {$\biggr)$} ;
	
	\path[backbone,on layer=background]
		(v0) edge[shorten >=-2pt] (3)
		(4) edge[shorten >=-2pt,shorten <=-4pt] (5);
	\draw[jumppointer]
		(v0) to[out=45,looseness=.4,in=90] (6);

	\begin{scope}[shift={(-.3,-2.75)}]
		\node (p1)  {${}+(1-p)\cdot {}$};
		
		\node[sn,xshift=.8em] (n2) at (p1.east) {$v_0$};
		\node (n3)  [right of = n2, xshift=.7em] {$\proc{RestDel} \biggl($};
		\node (n4) [new,xshift=1.2em] at (n3.east) {};
		\node[sn] (nJ1) at (n3.east) {\phantom{xxxxxxx}}; 
		\node (n5) at (nJ1.east){$\biggl)$ };
		\node[sn] (n6) [right of = n5, xshift=-1em]  {$v_j$};
		\node[sn,anchor=west,xshift=-.6pt] (nJ2) at (n6.east) {\phantom{xxxxxx}};

		\node[sn,yshift=8.5ex] (a2) at (n2.west) {$v_1$};
		\node (a3)  [right of = a2, xshift=.7em] {$\proc{RestDel} \biggl($};
		\node  (a4) [new,xshift=0em] at (a3.east) {};
		\node[sn] (aJ1) at (a3.east) {\phantom{xxxxxxx}}; 
		\node (a5) at (aJ1.east){$\biggl)$};
		\node[sn] (a6) [right of = a5, xshift=-1em]  {$v_j$};
		\node[sn,xshift=-.8pt] (aJ2) at (a6.east) {\phantom{xxxxxx}};

		\node[sn,yshift=-8.5ex] (b2) at (n2.west) {$v_0$};
		\node[sn,xshift=.7em] (bJ1) at (b2.east) {\phantom{xxxxxxx}}; 
		\node[xshift=.7em] (b3) at (bJ1.east) {$\proc{restDel} \biggl($};
		\node[sn] (b4) at (b3.east) {$v_j$};	
		\node[new,xshift=1em] (b5)  at (b4.east) {};
		\node[sn,anchor=west,xshift=-.8pt] (bJ2) at (b4.east) {\phantom{xxxxxx}}; 
		\node (b6) at (bJ2.east){$\biggl)$};
		
		\path[backbone,on layer=background]
			(n2) edge[shorten >=-2pt] (n3)
			(n5) edge[shorten <=-4pt] (n6)
			(a2) edge[shorten >=-2pt] (a3)
			(a5) edge[shorten <=-4pt] (a6)
			(b2) edge (bJ1)
			(bJ1) edge[shorten >=-2pt] (b3);
		\draw[jumppointer] (n2) to[out=45,looseness=.4,in=120] (n6);
		\draw[jumppointer] (a2) to[out=45,looseness=.4,in=120] (a6);
		\draw[jumppointer] (b2) to[out=45,looseness=.45,in=90] (b4);
	
		\draw[semithick,decorate, decoration={brace,amplitude=5pt}] 
			($(b2.south west)+(-.4em,-1ex)$) -- ($(a2.north west)+(-.4em,1ex)$);	
	\end{scope}
		
	\begin{scope}[shift={(8,-2.75)}, node distance = 0.8cm]
		\node (cm1) {$\Li={}$};
		\node (cm2) [sn] at (cm1.east) {$v_0$};
		\node (cm3) [new,xshift=2.2em] at (cm2.east) {};
		\node[sn,xshift=1em] (cmJ1) at (cm2.east) {\phantom{xxxxxxx}};
		\node[sn,xshift=1em] (cm4) at (cmJ1.east) {$v_j$};
		\node[sn,anchor=west,xshift=-0.8pt] (cmJ2) at (cm4.east) {\phantom{xxxxxx}};
		
		\node[yshift=8.5ex] (ca1) at (cm1.west) {$\Li={}$};
		\node[sn,fill=red] (ca2)  at (ca1.east) {$v_0$};
		\node[sn,xshift=.7em,minimum width=4mm] (ca3) at (ca2.east) {$v_1$};
		\node[sn,xshift=.7em] (caJ1) at (ca3.east) {\phantom{xxxx}};
		\node[sn,xshift=.7em] (ca4) at (caJ1.east) {$v_j$};
		\node[sn,anchor=west,xshift=-0.8pt] (caJ2) at (ca4.east) {\phantom{xxxxxx}};
		
		\node[yshift=-8.5ex] (cb1) at (cm1.west) {$\Li={}$};
		\node[sn] (cb2) at (cb1.east) {$v_0$};
		\node[sn,xshift=1em] (cbJ1) at (cb2.east) {\phantom{xxxxxxx}};
		\node[sn,xshift=1em] (cb3) at (cbJ1.east) {$v_j$};
		\node[new,xshift=1em] (cb4) at (cb3.east) {};
		\node[sn,anchor=west,xshift=-0.8pt] (cbJ2) at (cb3.east) {\phantom{xxxxxx}};

		\path[backbone,on layer=background]
			(cm2) edge (cmJ1.west)
			(cmJ1.east) edge (cm4)
			(ca2) edge (ca3)
			(ca3) edge (caJ1)
			(caJ1) edge (ca4)
			(cb2) edge (cbJ1.west)
			(cbJ1.east) edge (cb3);
		\draw[jumppointer] (cm2) to[out=45,looseness=.55,in=120] (cm4);
		\draw[jumppointer] (ca2) to[out=50,looseness=.5,in=120] (ca4);	
		\draw[jumppointer] (cb2) to[out=45,looseness=.5,in=120] (cb3);	
	\end{scope}
\end{tikzpicture}%
}}
%
%
%
%
%
%

\subsection{Pseudocode}
\label{sec:pseudocode}

We give full pseudocode for all basic operations on median-of-$k$ jumplists with
leaf size $w$ in this section.

We first list the four procedures \proc{Contains}, \proc{Insert}, \proc{Delete}
and \proc{RankSelect} that constitute the public interface of the data structure;
the other procedures can be thought of as low-level procedures typically
hidden from the user of the data structure.

We assume that jumplists are represented using the following records/objects.
\begin{codebox}
\Procname{Used Objects/Structs}
\li	$\mathrm{JumpList}(\id{head},\id{size})$
\li	$\mathrm{PlainNode}(\id{key},\id{next})$
\li	$\mathrm{JumpNode}(\id{key},\id{next},\id{jump},\id{nsize})$
\end{codebox}
References/pointers to nodes can refer to a PlainNode or to a JumpNode, 
and we assume there is an efficient method to check which type a particular instance has.
If $\id{node}$ is a reference to a PlainNode, we write $\id{node}.\id{key}$ and $\id{node}.\id{next}$
for the key-value and next-pointer fields of the referenced PlainNode;
similarly for the other types.

\begin{codebox}
	\Procname{$\proc{Contains}(\id{jumpList},x)$}
	\zi	\Comment Returns whether $x$ is present in \id{jumpList}
		and how many elements $<x$ it stores.
	\li	$(\id{node},r) \gets \proc{SpineSearch}(\id{jumpList}.\id{head},x)$
	\li	$\id{candidate} \gets \id{node}.\id{next}$ \label{lin:contains-null-check}
	\li	\Return $\bigl(\id{candidate}.\id{key}\isequal x,\, r \bigr)$
\end{codebox}

\begin{codebox}
	\Procname{$\proc{Insert}(\id{jumpList},x)$}
	\zi	\Comment Insert $x$ into \id{jumpList}; does nothing if $x$ is already present.
	\li	$(\id{node},r) \gets \proc{SpineSearch}(\id{jumpList}.\id{head},x)$
	\li	\If $\id{node}.\id{next}.\id{key} \ne x$ \;\Comment $x$ not yet present
			\label{lin:insert-null-check}
		\Then
	\li		$\id{node}.\id{next} \gets \New \mathrm{PlainNode}(x,\id{node}.\id{next})$
			\; \Comment Add new node in backbone.
	\li		$n \gets \id{jumpList}.\id{size} + 1$; \;
			$\id{jumpList}.\id{size} \gets n$
	\li		$\id{jumpList}.\id{head} \gets 
				\proc{RestoreAfterInsert}(\id{jumpList}.\id{head},n+1,r+1)$
		\EndIf
\end{codebox}

\begin{codebox}
	\Procname{$\proc{Delete}(\id{jumpList},x)$}
	\zi	\Comment Removes $x$ from \id{jumpList}; does nothing if $x$ is not present.
	\li	$(\id{node},r) \gets \proc{SpineSearch}(\id{jumpList}.\id{head},x)$
	\li	\If $\id{node}.\id{next}.\id{key} \isequal x$ \;\Comment $x$ is present
			\label{lin:delete-null-check}
		\Then
	\li		$\id{delNode} \gets \id{node}.\id{next}$
	\li		$\id{node}.\id{next} \gets \id{delNode}.\id{next}$
			\; \Comment Remove \id{delNode} from backbone.
	\li		$n \gets \id{jumpList}.\id{size} - 1$; \;
			$\id{jumpList}.\id{size} \gets n$
	\li		$\id{jumpList}.\id{head} \gets 
				\proc{RestoreAfterDelete}(\id{jumpList}.\id{head},n+1,r+1,\id{delNode})$
		\EndIf
\end{codebox}

\begin{codebox}
	\Procname{$\proc{RankSelect}(\id{jumpList},\id{rank})$}
	\zi \Comment Returns the element with (zero-based) rank $\id{rank}$, 
		\ie, the $(\id{rank}+1)$st smallest element.
	\li $\id{head} \gets \id{jumpList}.\id{head}$; \; $r\gets \id{rank}+1$
	\li \Repeat
	\li		\If $r > \id{head}.\id{nsize}$
			\Then
	\li			$r \gets r - (\id{head}.\id{nsize} + 1)$; \; 
				$\id{head} \gets \id{head}.\id{jump}$
	\li		\Else
	\li			$r \gets r-1$; \; 
				$\id{head}\gets \id{head}.\id{next}$
			\EndIf
	\li		\If $r \isequal 0$ \kw{then} \Return $\id{head.\id{key}}$ \kw{end if}
	\li \Until $\id{head}$ is PlainNode
	\li \Repeat
	\li			$r \gets r-1$; \; 
				$\id{head}\gets \id{head}.\id{next}$
	\li	\Until $r\isequal 0$
	\li	\Return \id{head}
\end{codebox}

The above methods make use of the following internal procedures.
We give a spine search implementation
that is augmented to determine also the rank of the found element.
Using the rank makes the procedures to restore the distribution after insertions or deletions 
a bit more convenient to state, and also avoids re-doing key comparisons there.

The given implementation of \proc{SpineSearch}, \proc{Contains}, \proc{Insert}, and \proc{Delete} 
assume a sentinel node $\id{tail}$ at the end
of the linked list that has $\id{tail}.\id{key} = +\infty$, \ie, a value larger than any actual key
value; we do however not count $\id{tail}$ towards the $m$ nodes of a jumplist since $\id{tail}$
can be shared across all instances of jumplists. 
The sentinel may never be the target of any jump pointer.
We could avoid the need for the sentinel at the expense of a null-check
of the next pointer, before comparing the successor's key
(\wref{lin:spine-search-null-check} in \proc{SpineSearch},
\wref{lin:contains-null-check} in \proc{Contains},
\wref{lin:insert-null-check} in \proc{Insert}, and
\wref{lin:delete-null-check} in \proc{Delete}).
Since using the sentinel is a bit more efficient and gives more readable code, 
we stick to this assumption.

\begin{codebox}
	\Procname{$\proc{SpineSearch}(\id{head}, x)$}
	\zi \Comment Returns last node with key $<x$ and its zero-based rank,
	\zi	\Comment \ie, the number of nodes with key $<x$
	\li $\id{rank} \gets 0$; \; $\id{steppedOver} \gets0$; \;
		$\id{lastJumpedTo} \gets \id{head}$
	\li \Repeat \qquad\Comment BST-style search
	\li		\If $\id{head}.\id{jump}.\id{key} < x$
			\Then
	\li			$\id{rank} \gets \id{rank} + \id{head}.\id{nsize} + 1 + \id{steppedOver}$
	\li			$\id{head} \gets \id{head}.\id{jump}$
	\li			$\id{steppedOver} \gets 0$; $\id{lastJumpedTo} \gets \id{head}$
	\li		\Else
	\li			$\id{head}\gets \id{head}.\id{next}$; \;
				$\id{steppedOver} \gets \id{steppedOver}+1$
			\EndIf
	\li \Until $\id{head}$ is PlainNode
	\li $\id{head}\gets\id{lastJumpedTo}$
	\li	\While $\id{head}.\id{next}.\id{key} < x$ \;\Comment Linear search from \id{lastJumpedTo}
			\label{lin:spine-search-null-check}
		\Do
	\li		$\id{head} \gets \id{head}.\id{next}$; \;
			$\id{rank} \gets \id{rank}+1$
		\EndWhile
	\li	\Return $(\id{head},\id{rank})$
\end{codebox}

\begin{codebox}
	\Procname{$\proc{Rebalance}(\id{head},m)$}
	\zi \Comment Draws jump pointers in for $m$ nodes starting with \id{head} (inclusive)
	\zi \Comment according to the randomized jumplist distribution.
	\zi \Comment Returns new first (possibly still \id{head}) and last node of the sublist.
	\li	\If $m \le w$
		\Then
	\li		Replace \id{head} and its $m-1$ successors by $m$ linked PlainNodes.
	\li		\Return $(\text{new head},\text{new end})$
	\li \Else
	\li		$S \gets$ random $k$-element subset of $[2..m-1]$
	\li		$\id{jumpIndex} \gets \proc{Median}(S)$ \; \Comment $S_{(t_1+1)}$ in general
	\li		\Return $\proc{SetJumpAndRebalance}(\id{head},m,\id{jumpIndex})$
		\EndIf
\end{codebox}

\begin{codebox}
	\Procname{$\proc{SetJumpAndRebalance}(\id{head},m,j)$}
	\zi \Comment Rebalances the sublist starting at \id{head} containing $m$ nodes,
	\zi \Comment where we fix the topmost jump pointer to point to the element of rank $j$.
	\zi \Comment Returns new first (possibly still \id{head}) and last node of the sublist.
	\li	$(\id{nextStart},\id{nextEnd}) \gets \proc{Rebalance}(\id{head}.\id{next},j-1)$
	\li	$(\id{jumpStart},\id{jumpEnd}) \gets \proc{Rebalance}(\id{nextEnd}.\id{next},m-j)$
	\li	$\id{nextEnd}.\id{next} \gets \id{jumpStart}$
	\li \Return $\bigl(\New \mathrm{JumpNode}(\id{head}.\id{key}, \id{nextStart},\id{jumpStart},j-1 ), 
			\id{jumpEnd} \bigr)$
\end{codebox}

\begin{codebox}
	\Procname{$\proc{RestoreAfterInsert}(\id{head},m,r)$}
	\zi \Comment Restore distribution in sublist with header \id{head} and of size $m$ 
		after an insertion at position $r$.
	\zi \Comment $m$ is the number of nodes in the sublist, including the new element
	\zi \Comment Returns the new head of the sublist (possibly still \id{head}).
	\li	\If $m\le w+1$ \;\Comment Base case
		\Then
	\li		\If $m \isequal w+1$ \;\Comment We need a new JumpNode, so rebalance.
			\Then
	\li			$(\id{head},\id{end}) \gets \proc{Rebalance}(\id{head},m)$
			\EndIf
	\li	\Else
	\li		$\id{newElementInSample}\gets \proc{CoinFlip}\bigl(\frac{k}{m-2}\bigr)$
			\;\Comment \id{true} with probability $\frac{k}{m-2}$
	\li		\If \id{newElementInSample} \;\Comment Rebalance conditional on new index being in sample.
			\Then
	\li			$\id{newIndex} \gets \max\{r,2\}$
	\li			$S \gets$ random $(k-1)$-element subset of $[2..m-1]\setminus \{\id{newIndex}\}$
	\li			$\id{jumpIndex} \gets \proc{Median}(\id{newIndex} \cup S)$ 
					\; \Comment $S_{(t_1+1)}$ in general
	\li			$(\id{head},\id{end}) \gets \proc{SetJumpAndRebalance}(\id{head},m,\id{jumpIndex})$
	\li		\Else \;\Comment topmost jump pointer can be kept
	\li			\If $r \isequal 0$ \;\Comment new node is head of sublist, so steal successor's jump.
				\Then
	\zi				\Comment Swap roles of the two nodes.
	\li				$\id{succ} \gets \id{head}.\id{next}$
	\li				Swap $\id{key}$ and $\id{next}$ fields of \id{head} and \id{succ}.
	\li				$\id{head} \gets \id{succ}$
				\EndIf
	\li				$J_1 \gets \id{head}.\id{nsize}$
	\li			\If $r \le J_1 + 1$ \;\Comment New element is in next-sublist.
				\Then
	\li				$J_1 \gets J_1 + 1$; \; $\id{head}.\id{nsize} \gets J_1$
	\li				$\id{succ} \gets \proc{RestoreAfterInsert}(\id{head}.\id{next},J_1,\max\{0,r-1\})$
	\li				$\id{head}.\id{next} \gets \id{succ}$
	\li			\Else \;\Comment New element is in jump-sublist.
	\li				$J_2 \gets m-1-J_1$
	\li				$\id{jumpHead} \gets 
						\proc{RestoreAfterInsert}(\id{head}.\id{jump},J_2,r-1-J_1)$
	\li				\If $\id{jumpHead} \ne \id{head}.\id{jump}$ \;\Comment Have to reconnect backbone
					\Then
	\li					$\id{head}.\id{jump} \gets \id{jumpHead}$
	\li					$(\id{lastInNext},rr) \gets 
							\proc{SpineSearch}(\id{head},\id{jumpHead}.\id{key})$
	\li					$\id{lastInNext}.\id{next} \gets \id{jumpHead}$
					\EndIf 
				\EndIf
			\EndIf
		\EndIf
	\li	\Return \id{head}
\end{codebox}

\begin{codebox}
	\Procname{$\proc{RestoreAfterDelete}(\id{head},m,r,\id{delNode})$}
	\zi \Comment Restore distribution in sublist with header \id{head} and size $m$ 
		after a deletion at position $r$.
	\zi \Comment $m$ is the number of nodes in the sublist, 
		excluding the just deleted element $\id{delNode}$.
	\zi \Comment Returns the new head of the sublist (possibly still \id{head}).
	\li	\If $m\le w$
		\Then
	\li		\If $m \isequal w$ \;\Comment head is a JumpNode, must become PlainNode
			\Then
	\li			$\id{head} \gets \New \mathrm{PlainNode}(\id{head}.\id{key},\id{head}.\id{next})$
			\EndIf
	\li	\Else
	\li		$J_1 \gets 
				\begin{dcases*}
					\id{delNode}.\id{nsize}, & \kw{if} $r\isequal 0$;\\
					\id{head}.\id{nsize}, & \kw{else}.
				\end{dcases*}
			$ 
			\\[-.75\baselineskip]
	\li		$\like[l]{I_1}{p}\gets 
				\begin{dcases*}
					1\,, & \kw{if} $r \isequal J_1+1$; \;\Comment deleted node was jump target \\
					[J_1 \isequal 1 \land r\le 1], & \kw{else if} $k\isequal 1$; 
							\;\Comment special case to avoid $\frac00$ \\
					\frac{t_1}{J_1-1}, & \kw{else if} $r < J_1+1$;\\
					\frac{t_2}{m-1-J_1}, & \kw{else}.
				\end{dcases*}
			$
			\\[-.75\baselineskip]
	\li		$\id{deletedElementInSample}\gets \proc{CoinFlip}(p)$
			\;\Comment \id{true} with probability $p$
	\li		\If \id{deletedElementInSample} 
				\;\Comment Rebalance sublist.
			\Then
	\li			$(\id{head},\id{end}) \gets \proc{Rebalance}(\id{head},m)$
	\li		\Else \;\Comment Topmost jump pointer can be kept.
	\li			\If $r \isequal 0$ \;\Comment Impose deleted head's pointer onto successor.
				\Then
	\zi				\Comment Swap roles of the two nodes.
	\li				Swap $\id{key}$ and $\id{next}$ fields of \id{head} and \id{delNode}.
	\li				Swap $\id{head}$ and $\id{delNode}$.
				\EndIf
	\li			\If $r < J_1 + 1$ \;\Comment Deletion in next-sublist.
				\Then
	\li				$J_1 \gets J_1 - 1$; \; $\id{head}.\id{nsize} \gets J_1$
	\li				$\id{succ} \gets 
						\proc{RestoreAfterDelete}(\id{head}.\id{next},J_1,\max\{0,r-1\},\id{delNode})$
	\li				$\id{head}.\id{next} \gets \id{succ}$
	\li			\Else \;\Comment Deletion in jump-sublist.
	\li				$J_2 \gets m-1-J_1$
	\li				\rlap{$\id{jumpHead} \gets 
						\proc{RestoreAfterDelete}(\id{head}.\id{jump},J_2,r-1-J_1,\id{delNode})$}
	\li				\If $\id{jumpHead} \ne \id{head}.\id{jump}$ \;\Comment Have to reconnect backbone
					\Then
	\li					$\id{head}.\id{jump} \gets \id{jumpHead}$
	\li					$(\id{lastInNext},rr) \gets 
							\proc{SpineSearch}(\id{head},\id{jumpHead}.\id{key})$
	\li					$\id{lastInNext}.\id{next} \gets \id{jumpHead}$
					\EndIf 
				\EndIf
			\EndIf
		\EndIf
	\li	\Return \id{head}
\end{codebox}

%% file: jumplists.bbl
\begin{thebibliography}{10}

\bibitem{AnderssonFagerbergLarsen2005}
A.~Andersson, R.~Fagerberg, and K.S. Larsen.
\newblock Balanced binary search trees.
\newblock In D.~Mehta and S.~Sahni, editors, {\em Handbook of Data Structures
  and Applications}, chapter~10. CRC Press, 2005.

\bibitem{BronnimannCazalsDurand2003}
Herv\'e Br{\"{o}}nnimann, Fr\'ed\'eric Cazals, and Marianne Durand.
\newblock Randomized jumplists: A jump-and-walk dictionary data structure.
\newblock In {\em STACS 2003}, pages 283--294, 2003.
\newblock \href {http://dx.doi.org/10.1007/3-540-36494-3_26}
  {\path{doi:10.1007/3-540-36494-3_26}}.

\bibitem{CasasDiazMartinez1991}
R.~Casas, J.~D{\'{\i}}az, and C.~Martinez.
\newblock Statistics on random trees.
\newblock In {\em International Colloquium on Automata, Languages, and
  Programming (ICALP)}, pages 186--203. Springer, 1991.
\newblock \href {http://dx.doi.org/10.1007/3-540-54233-7_134}
  {\path{doi:10.1007/3-540-54233-7_134}}.

\bibitem{DeanJones2007}
Brian~C. Dean and Zachary~H. Jones.
\newblock Exploring the duality between skip lists and binary search trees.
\newblock In {\em Annual southeast regional conference}, pages 395--399. {ACM}
  Press, 2007.
\newblock \href {http://dx.doi.org/10.1145/1233341.1233413}
  {\path{doi:10.1145/1233341.1233413}}.

\bibitem{Drmota2009}
Michael Drmota.
\newblock {\em Random Trees}.
\newblock Springer, 2009.

\bibitem{Elmasry2005}
Amr Elmasry.
\newblock Deterministic jumplists.
\newblock {\em Nordic Journal of Computing}, 12(1):27--39, 2005.

\bibitem{Flajolet2009}
Philippe Flajolet and Robert Sedgewick.
\newblock {\em Analytic Combinatorics}.
\newblock Cambridge University Press, 2009.
\newblock URL: \url{http://algo.inria.fr/flajolet/Publications/book.pdf}.

\bibitem{ConcreteMathematics}
Ronald~L. Graham, Donald~E. Knuth, and Oren Patashnik.
\newblock {\em Concrete Mathematics: A Foundation For Computer Science}.
\newblock Addison-Wesley, 1994.

\bibitem{Greene1983}
Daniel~Hill Greene.
\newblock {\em {Labelled formal languages and their uses}}.
\newblock {Ph.\hspace{.2ex}D.} thesis, Stanford University, 1983.

\bibitem{Hennequin1989}
Pascal Hennequin.
\newblock Combinatorial analysis of {Q}uicksort algorithm.
\newblock {\em RAIRO - Theoretical Informatics and Applications},
  23(3):317--333, 1989.

\bibitem{Hennequin1991}
Pascal Hennequin.
\newblock {\em Analyse en moyenne d'algo\-rith\-mes~: tri rapide et arbres de
  recherche}.
\newblock Th\`ese ({P}h.\,{D.} {T}hesis), Ecole Politechnique, Palaiseau, 1991.

\bibitem{HuangWong1983}
Shou-Hsuan~Stephen Huang and C.~K. Wong.
\newblock Binary search trees with limited rotation.
\newblock {\em BIT}, (4):436--455, 1983.
\newblock \href {http://dx.doi.org/10.1007/BF01933619}
  {\path{doi:10.1007/BF01933619}}.

\bibitem{HuangWong1984}
Shou-Hsuan~Stephen Huang and C.~K. Wong.
\newblock Average number of rotations and access cost in {iR}-trees.
\newblock {\em BIT}, 24(3):387--390, 1984.
\newblock \href {http://dx.doi.org/10.1007/BF02136039}
  {\path{doi:10.1007/BF02136039}}.

\bibitem{Knuth1998}
Donald~E. Knuth.
\newblock {\em {The Art Of Computer Programming: Searching and Sorting}}.
\newblock Addison Wesley, 2nd edition, 1998.

\bibitem{Mahmoud1992evolution}
Hosam~M. Mahmoud.
\newblock {\em Evolution of Random Search Trees}.
\newblock Wiley, 1992.

\bibitem{MartinezRoura1998}
Conrado Mart\'{\i}nez and Salvador Roura.
\newblock Randomized binary search trees.
\newblock {\em J. ACM}, 45(2):288--323, 1998.
\newblock \href {http://dx.doi.org/10.1145/274787.274812}
  {\path{doi:10.1145/274787.274812}}.

\bibitem{MunroPapadakisSedgewick1992}
J.~Ian Munro, Thomas Papadakis, and Robert Sedgewick.
\newblock Deterministic skip lists.
\newblock In {\em ACM-SIAM Symposium on Discrete Algorithms}, SODA 1992, pages
  367--375. SIAM, 1992.

\bibitem{Neumann2015}
Elisabeth Neumann.
\newblock {\em Randomized Jumplists With Several Jump Pointers}.
\newblock Bachelor's thesis, 2015.
\newblock URL:
  \url{http://nbn-resolving.de/urn/resolver.pl?urn:nbn:de:hbz:386-kluedo-41642}.

\bibitem{NievergeltReingold1973}
J.~Nievergelt and E.~M. Reingold.
\newblock Binary search trees of bounded balance.
\newblock {\em SIAM Journal on Computing}, 2(1):33--43, 1973.
\newblock \href {http://dx.doi.org/10.1137/0202005}
  {\path{doi:10.1137/0202005}}.

\bibitem{Pasanen2010}
Tomi~A. Pasanen.
\newblock Random binary search tree with equal elements.
\newblock {\em Theoretical Computer Science}, 411(43):3867--3872, 2010.
\newblock \href {http://dx.doi.org/10.1016/j.tcs.2010.06.023}
  {\path{doi:10.1016/j.tcs.2010.06.023}}.

\bibitem{PobleteMunro1985}
Patricio~V Poblete and J.~Ian Munro.
\newblock The analysis of a fringe heuristic for binary search trees.
\newblock {\em Journal of Algorithms}, 6(3):336--350, 1985.
\newblock \href {http://dx.doi.org/10.1016/0196-6774(85)90003-3}
  {\path{doi:10.1016/0196-6774(85)90003-3}}.

\bibitem{Pugh1990}
William Pugh.
\newblock Skip lists: A probabilistic alternative to balanced trees.
\newblock {\em Communications of the ACM}, 33(6):668--676, 1990.
\newblock \href {http://dx.doi.org/10.1145/78973.78977}
  {\path{doi:10.1145/78973.78977}}.

\bibitem{Roura2001}
Salvador Roura.
\newblock {Improved Master Theorems for Divide-and-Conquer Recurrences}.
\newblock {\em Journal of the ACM}, 48(2):170--205, 2001.

\bibitem{SeidelAragon1996}
R.~Seidel and C.~R. Aragon.
\newblock Randomized search trees.
\newblock {\em Algorithmica}, 16(4-5):464--497, 1996.
\newblock URL: \url{http://link.springer.com/10.1007/BF01940876}, \href
  {http://dx.doi.org/10.1007/BF01940876} {\path{doi:10.1007/BF01940876}}.

\bibitem{Walker1976}
A.~Walker and D.~Wood.
\newblock Locally balanced binary trees.
\newblock {\em The Computer Journal}, 19(4):322--325, 1976.
\newblock \href {http://dx.doi.org/10.1093/comjnl/19.4.322}
  {\path{doi:10.1093/comjnl/19.4.322}}.

\bibitem{Wild2016}
Sebastian Wild.
\newblock {\em Dual-Pivot Quicksort and Beyond: Analysis of Multiway
  Partitioning and Its Practical Potential}.
\newblock Doktorarbeit (\phdthesis), Technische Universit{\"a}t Kaiserslautern,
  2016.
\newblock URL:
  \url{http://nbn-resolving.de/urn/resolver.pl?urn:nbn:de:hbz:386-kluedo-44682}.

\bibitem{codeWild2016}
Sebastian Wild.
\newblock sebawild/jumplists: snapshot-for-paper.
\newblock 2016.
\newblock \href {http://dx.doi.org/10.5281/zenodo.155326}
  {\path{doi:10.5281/zenodo.155326}}.

\bibitem{Wild2018}
Sebastian Wild.
\newblock Quicksort is optimal for many equal keys.
\newblock In {\em Workshop on Analytic Algorithmics and Combinatorics
  ({ANALCO})}, pages 8--22. SIAM, 2018.
\newblock \href {http://dx.doi.org/10.1137/1.9781611975062.2}
  {\path{doi:10.1137/1.9781611975062.2}}.

\end{thebibliography}
